\documentclass[journal]{IEEEtran}
\onecolumn
\usepackage{mathrsfs}
\usepackage{blkarray}
\usepackage{amsmath}
\usepackage{amsthm}
\usepackage{graphicx}
\usepackage{amssymb}
\usepackage{enumerate}
\usepackage{longtable,tabularx,float}
\usepackage{cite}
\usepackage{mathcomp}
\usepackage{supertabular}
\usepackage{stmaryrd}
\usepackage{color}
\usepackage{url}
\usepackage{algorithm}
\usepackage{algpseudocode}
\usepackage{bm}
\usepackage{hyperref}

\newtheorem{lemma}{Lemma}[section]
\newtheorem{theorem}{Theorem}[section]
\newtheorem*{theorem*}{Theorem}
\newtheorem*{lemma*}{Lemma}

\newtheorem{example}{Example}[section]
\newtheorem{construction}{Construction}[section]
\newtheorem{remark}{Remark}[section]

\newcommand{\A}{\mathcal {A}}
\newcommand{\B}{\mathcal {B}}
\newcommand{\mm}{{~\text{mod}^+~}}
\newcommand{\Q}{\mathcal {Q}}
\renewcommand{\P}{\mathcal {P}}
\renewcommand{\S}{\mathcal {S}}

\newcommand{\bbF}{{\mathbb F}}

\begin{document}

\title{Sparse and Balanced MDS Codes over Small Fields}
\author{Tingting~Chen,~and~Xiande~Zhang
        \thanks{T. Chen ({\tt ttchenxu@mail.ustc.edu.cn}) is with
        School of Cyber Security, University of Science and Technology of China, Hefei, 230026, Anhui, China.}

\thanks{X. Zhang ({\tt drzhangx@ustc.edu.cn}) is with School of Mathematical Sciences,
University of Science and Technology of China, Hefei, 230026, Anhui, China.}}

%
\maketitle

\begin{abstract}
Maximum Distance Separable (MDS) codes with a sparse and balanced generator matrix are appealing  in distributed storage systems for balancing and minimizing the computational load. Such codes have been constructed via Reed-Solomon codes over large fields. In this paper, we focus on small fields. We prove that there exists an $[n,k]_q$ MDS code that has a sparse and balanced generator matrix for any $q\geq n$ provided that $n\leq 2k$,
 by designing several algorithms with complexity running in polynomial time in $k$ and $n$.
\end{abstract}


\section{Introduction}
\IEEEPARstart{M}{DS} codes with constrained generator matrices have been attracting much attention recently due to their applications in weakly secure cooperative data exchange \cite{yan2013algorithms,yan2014weakly,li2017cooperative}, multiple access networks \cite{halbawi2014distributed,dau2015simple}, wireless
sensor networks \cite{Dau2013}, and so on. The relations among them are well explained in \cite{DauSY14,YildizH19}.  An interesting problem of this topic is to construct an $[n,k]_q$ MDS code with a sparse and balanced generator matrix (SBGM) $G$, where `sparse' means that each row of $G$ has the least possible number of nonzeros, i.e., $n-k+1$ nonzeros, and `balanced' means that the numbers of nonzeros in any two columns differ by at most one, i.e., $n-\lceil\frac{k(k-1)}{n}\rceil$ or  $n-\lfloor\frac{k(k-1)}{n}\rfloor$. This problem was first considered in \cite{Dau2013}. Such a matrix gives us some benefits during the encoding process \cite{halbawi2018sparse,HalbawiLH16}. On the one hand, since the time required to compute each code symbol is a function of the number of nonzeros in a specified column of $G$, each code symbol is computed in roughly the same amount of time due to the balanced property of  $G$. This ensures that the computational load is balanced, which is required in  scenarios such as the storage system. On the other hand, when $G$ is sparse, then updating a single message symbol impacts exactly $n-k+1$ storage nodes in the storage system.

In the recent few years, progress has been reported on the above problem. In \cite{Dau2013}, it was shown that there always exists an MDS code with an SBGM over any finite field of size $q>\binom{n-1}{k-1}$ through a probabilistic argument. The authors in \cite{HalbawiLH16,halbawi2016balanced} constructed an $[n, k]_q$ cyclic Reed-Solomon code that has an SBGM for any prime power $q=n+1$ and any $k$ such that $1\leq k\leq n$. Song and Cai \cite{SongC18} further extended their results by proving that for any positive integers $n$ and $k$ such that $1\leq k\leq n$, there exists an $[n, k]_q$ generalized Reed-Solomon code that has an SBGM over any finite field $\bbF_q$ of size $q\geq n+\lceil\frac{k(k-1)}{n}\rceil$. But there is still a gap between the code length and the field size when $k\geq 3$. It is natural to ask  whether there exists an $[n,k]_q$ MDS code with an SBGM over a smaller field with $q<n+\lceil\frac{k(k-1)}{n}\rceil$. Motivated by this problem, we focus on constructions of $[n,k]_q$ MDS codes with an SBGM for all $q\geq n$ and $k\geq 3$ in this paper.

\subsection{Related Work}

One of the challenging problems referring to MDS codes is the well known \emph{MDS conjecture}, which states that there exists an $[n,k]_q$ MDS code if and only if $n\leq q+1$ for all $q$ and $2\leq k\leq q-1$, except when $q$ is even and $k\in\{3,q-1\}$, in which case $n\leq q+2$. The sufficiency of the MDS conjecture has been proved via the use of (extended) Generalized Reed-Solomon codes \cite{macwilliams1977theory}.

 The problem of MDS codes with support constrained generator matrices is asking whether an MDS code exists with a prescribed zero patterns in the generator matrix. This problem has been studied in \cite{DauSY14,YildizH19,GreavesS19} for MDS codes with Hamming distance and in \cite{Yildiz2020IEEE,Yildiz2020ISIT} for MDS codes with rank metric (Gabidulin codes).
 Let $G$ be the $k\times n$ generator matrix of an $[n,k]_q$ MDS code with $k\leq n$. Define the \emph{zero pattern} of $G$ as a set system $S_1,\ldots,S_k\subset \{1,2,\ldots,n\}$, where $S_i=\{j\leq n: G_{i,j}=0\}$. The necessary condition of this set
system is known as the {\it MDS condition}: $|I|+|\cap_{i\in I}S_i|\leq k$ for any nonempty $I\subseteq \{1,2,\ldots,k\}$. It is conjectured that the MDS condition is sufficient for the existence of MDS codes whose generator matrices have the given zero pattern when $q\geq n+k-1$ \cite{DauSY14}.
This conjecture is known as the \emph{GM-MDS Conjecture}, and attracts a lot of interest, see \cite{effros2015between,heidarzadeh2017algebraic,yildiz2018further}. Recently the GM-MDS Conjecture was proved to be true by Lovett \cite{lovett2018mds} and independently by Yildiz and Hassibi \cite{YildizH19}, which we restated as a theorem as follows.

\begin{theorem}[{\cite[GM-MDS Theorem]{DauSY14,YildizH19,lovett2018mds}}]
  Let $\mathcal{S}=\{S_1,\ldots,S_k\}$ be a set system where $S_i\subseteq \{1,2,\ldots,n\}$, $1\leq i \leq k$. Then for  $q\geq n+k-1$, there exists an $[n,k]_q$ MDS code with a generator matrix $G$ over $\bbF_q$ such that $G_{i,j}=0$ whenever $j\in S_i$, if and only if $\mathcal{S}$ satisfies the MDS condition.
\end{theorem}

Further results on the existence of MDS codes with slightly stronger support constraint than the MDS condition on the generator matrices but with field size $q\geq n$ or $q\geq n+1$ are considered in \cite{GreavesS19}. However, their results can not be used to give sparse and balanced MDS codes.
The authors in \cite{GreavesS19} gave two constructions of some special classes of $[n,k]_q$ Reed-Solomon codes  whose generator matrices have constrained support. One of them is that over any finite field $\bbF_q$ with  $q\geq n$,  there exists an $[n,k]_q$ Reed-Solomon code if the zero pattern $\mathcal{S}=\{S_1,\ldots,S_k\}$ of its generator matrix satisfies the MDS condition, and further $|\cap_{j=1}^iS_j|=k-i$ for all $1\leq i \leq k$.  Notice that when $i=k-1$, it is required that $|\cap_{j=1}^{k-1}S_j|=1$. This means there is at least one column of the generator matrix $G$ containing $k-1$ zeros.
  So only when $\lceil\frac{k(k-1)}{n}\rceil\geq k-1$, that is, $n=k$ or $k+1$, $G$  can be sparse and balanced.
The second construction  \cite{GreavesS19} they gave is that over any finite field $\bbF_q$ with  $q\geq n+1$, there exists an $[n,k]_q$ Reed-Solomon code if the zero pattern $\mathcal{S}=\{S_1,\ldots,S_k\}$ of its generator matrix satisfies $|S_i|\leq i-1$ for all $i=1,\ldots,k$. Notice that when $i\leq k-1$, $|S_i|\leq k-2$, the generator matrix $G$ obtained from this construction is not sparse.




\subsection{Our Contribution}

In this paper, we construct an $[n,k]_q$ MDS code with an SBGM $G$ over any finite field of size $q\geq n$ when $3\leq k\leq n\leq 2k$. It suffices to find a $k\times n$ matrix $G$ over $\mathbb{F}_q$ with a sparse and balanced zero pattern, such that all minors of $G$ have full rank.  We first give a sufficient condition for the existence of a sparse generator matrix $G$ described by the set system $\mathcal{S}=\{S_1,\ldots,S_k\}$, see Theorem~\ref{thm_MDS}, which extends \cite[Theorem~II.5]{GreavesS19}. Then we show that  the set system $\mathcal{S}$ satisfying the sufficient condition in Theorem~\ref{thm_MDS} is balanced only if $n\leq 2k$. Finally, the binary matrix corresponding to $\mathcal{S}$ is proved to exist whenever $n\leq 2k$, by several algorithms.  We state our main result in the following theorem.
\begin{theorem}\label{thm_main} For any integer $k\geq 3$, let $n\leq 2k$ if $k$ is even, and $n\leq 2k-1$ if $k$ is odd.
  For any finite field $\mathbb{F}_q$ with $q \geq n$, there exists an $[n,k]_q$ MDS code, whose generator matrix is sparse and balanced.
\end{theorem}

The proof of Theorem~\ref{thm_main} is completed by designing several algorithms that have complexity running in polynomial time in $k$ and $n$. These algorithms output a sparse and balanced binary matrix, which is the complement of the incidence matrix of a set system $\mathcal{S}$ satisfying  Theorem~\ref{thm_MDS}.

\subsection{Organization}

This paper is organized as follows. In Section~\ref{pre}, necessary notations and definitions are given first, and then a sufficient condition on the zero pattern of a sparse generator matrix of an $[n,k]_q$ MDS code  with $q\geq n$ is provided. Details of constructions of  balanced zero patterns satisfying the sufficient condition are given in Section~\ref{conca}, which rely on several key operations on matrices. Finally, a brief conclusion is given in Section~\ref{con}.

\section{Support Constraints of MDS Codes}\label{pre}
We start by introducing some basic notations and definitions, and then proceed to the sufficient condition on the zero pattern of a sparse generator matrix of an $[n,k]_q$ MDS code.

\subsection{Notations and Definitions}

For any integers $a<b$, let $[a,b]$ denote the set of integers $\{a,a+1,\ldots,b\}$. We further abbreviate $[1,b]$ as $[b]$. Let $a\mm n$ denote the unique $r\in[n]$ such that $n$ divides $a-r$, and let $[a,b]\mm n$ denote the set $\{x\mm n: x\in[a,b]\}$. 

 We use $\bbF_q$ to denote the finite field with $q$ elements. A linear code $\mathcal{C}$ over $\bbF_q$ of length $n$, dimension $k$ and minimum distance $d$ is denoted by $[n,k,d]_q$. When $\mathcal{C}$ is an MDS code, i.e., $d=n-k+1$, we sometimes omit $d$ and write $[n,k]_q$.
 A generator matrix $G$ of  $\mathcal{C}$ is said to be \emph{sparse and balanced} \cite{SongC18} if $G$ satisfies the following two conditions:
 \begin{itemize}
 \item[(1)]
Sparse condition: the weight of each row of $G$ is
exactly $n-k + 1$;
\item[(2)] Balanced condition: the weight of each column of $G$ is
either $\lceil\frac{k(n-k+1)}{n}\rceil$ or  $\lfloor\frac{k(n-k+1)}{n}\rfloor$.
\end{itemize}
An MDS code that has a sparse and balanced generator matrix
(SBGM) is simply called a sparse and balanced MDS code. In this paper, we focus on constructions of sparse and balanced Reed-Solomon codes. An $[n,k]_q$ {\it Reed-Solomon} (RS) code is a special MDS code, which is given by $\{(f(a_1),\ldots,(f(a_n)):f\in\bbF_q[x], \text{deg}(f)<k\}$, where the evaluation points $a_1,\ldots,a_n\in\bbF_q$ are all distinct.

 Let $\P$ be a  sequence of $k$ polynomials  $f_1,f_2,\ldots,f_k$ in $\bbF_q[x]$ such that $\deg(f_i)\leq t-1$, then the coefficient matrix $C(\P)$  of $\P$ is a $k\times t$ matrix with the $(i,j)$th entry being the coefficient of $x^{t-j}$ in $f_i$, that is $[x^{t-j}]f_i$. If $\P$ consists of only one polynomial $f$, we simply write $C(f)$ as the row vector recording all coefficients of $f$.

Given a set system $\mathcal{S}=\{S_1,S_2,\ldots,S_k\}$  with $S_i\subset [n]$ and $|S_i|\leq k-1$ for each $i\in[k]$.  Let $a_1,\ldots,a_n$ be any fixed $n$ distinct elements of the field $\mathbb{F}_q$. Define
 $P_{S_i}(x)\triangleq\prod_{j\in S_i}(x-a_j)=p_{i,0}x^{k-1}+p_{i,1}x^{k-2}+\cdots+p_{i,k-1}\in\bbF_q[x]$, for all $i\in[k]$. In the rest of this paper, we denote $\P$  the sequence $P_{S_1},P_{S_2},\ldots,P_{S_k}$, then $C(\P)=(p_{i,j})_{k\times k}$ for $i\in[k]$ and $j\in [0,k-1]$. Let $G=(g_{i,j})$ be the $k\times n$ matrix over $\mathbb{F}_q$ with $g_{i,j}=P_{S_i}(a_j)$ for $i\in [k]$ and $j\in [n]$. Then
   \begin{align*}
    G&=\begin{bmatrix}
      p_{1,0}& p_{1,1}&\cdots&p_{1,k-1}\\
      p_{2,0}& p_{2,1}&\cdots&p_{2,k-1}\\
      p_{k,0}& p_{k,1}&\cdots&p_{k,k-1}
    \end{bmatrix}
    \begin{bmatrix}
       a_1^{k-1}&a_2^{k-1} &\cdots& a_n^{k-1}\\
      a_1^{k-2}&a_2^{k-2} &\cdots& a_n^{k-2}\\
      a_1^0&a_2^0 &\cdots& a_n^0\\
    \end{bmatrix}\\
    &=C(\P)\cdot V,
  \end{align*}
  where $V$ is the Vandermonde matrix.  It is easy to check that if $P_{S_1},P_{S_2},\ldots,P_{S_k}$ are linearly independent over $\bbF_q$, then $\det (C(\P))\neq 0$, and hence any $k$ columns of $G$ are linearly independent, so $G$ can be seen as a generator matrix of an $[n,k]_q$ RS code with the evaluation points $a_1,\ldots,a_n$. In other words, to construct an $[n,k]_q$ MDS code, one would like to construct a set system $\mathcal{S}=\{S_1,S_2,\ldots,S_k\}$, such that the polynomials $P_{S_1},P_{S_2},\ldots,P_{S_k}$ defined by $\mathcal{S}$ are linearly independent.


Let $M_{\mathcal{S}}=(m_{i,j})$ be the $k\times n$ binary matrix with $m_{i,j}=0$ if and only if $j\in S_i$.  Then $m_{i,j}=0$ if and only if $g_{i,j}=0$, so we  call $M_{\mathcal{S}}$ the {\it complementary support matrix} of $G$, or of $\S$.
Given a generator matrix $G$ of an $[n,k]_q$ MDS code, we can determine its complementary support matrix, and then obtain a set system $\mathcal{S}=\{S_1,S_2,\ldots,S_k\}$ with $S_i=\{j\in [n]:g_{i,j}=0\}\subset [n]$ for each $i\in [k]$.  The size of each $S_i$ is at most $k-1$ by the minimum distance $d=n-k+1$, and $|S_i|=k-1$ for all $i$ if $G$ is sparse.

Next, we show that  a set system $\mathcal{S}=\{S_1,S_2,\ldots,S_k\}$ with certain properties will produce $k$ linearly independent polynomials $\P$ over  $\mathbb{F}_q$ with $q\geq n$, that is $\det (C(\P))\neq 0$, and consequently a generator matrix $G$ for an $[n,k]_q$ RS code.

\subsection{Support Constraints of Sparse Codes}


Assume that $n\geq k>1$ in this section.
  Let $\mathcal{S}=\{S_1,S_2,\ldots,S_k\}$ be a $(k-1)$-\emph{uniform} set system over $[n]$, that is $|S_i|=k-1$ for each $i\in[k]$. If there exists $i\in[k-1]$, such that $|S_1\cap S_2\cap\cdots\cap S_i|=k-i$ and $|S_{i+1}\cap S_{i+2}\cap\cdots\cap S_k|=i$, then we call $\mathcal{S}$ is \emph{separable at $i$}. If $\cap_{i\in[k]}S_i =\emptyset$, we say $\mathcal{S}$ is {\it non-intersecting}.


  Given a non-intersecting set system $\mathcal{S}=\{S_1,S_2,\ldots,S_k\}$ which is separable at some $i\in [k-1]$, assume that $A=S_1\cap S_2\cap\cdots\cap S_i$ is of size $k-i$ and $B=S_{i+1}\cap S_{i+2}\cap\cdots\cap S_k$ is of size $i$. By the non-intersecting property, $A\cap B=\emptyset$. Denote $S_j'=S_j\setminus A$ for $j\in[i]$, and $S_j'=S_j\setminus B$ for $j\in[i+1,k]$.  We say $\A=\{S_1',\ldots,S_i'\}$ and $\B=\{S_{i+1}',\ldots,S_k'\}$ are  the two \emph{residual} set systems of $\mathcal{S}$ with index $i$. Note that $\A$ is $(i-1)$-uniform and $\B$ is $(k-i-1)$-uniform. In particular, when $i=1$ or $i=k-1$, then $\A$ or $\B$ will degenerate into $\{\emptyset\}$. Let $\P_1$ be the sequence of polynomials $P_{S_1'},P_{S_2'},\ldots,P_{S_i'}$ defined by $\A$, and $\P_2$ be the sequence  $P_{S_{i+1}'},\ldots,P_{S_k'}$ defined by $\B$. Here, if the set $S'=\emptyset$, we simply define $P_{S'}(x)=1$. Denote $f_0=\prod_{u\in A}(x-a_{u})$ and $g_0=\prod_{v\in B}(x-a_{v})$. Finally, let $\Q_1$ be the sequence $x^{i-1}f_0$, $x^{i-2}f_0$, $\ldots$, $f_0$ and $\Q_2$ be the sequence $x^{k-i-1}g_0$, $x^{k-i-2}g_0$, $\ldots$, $g_0$. Note that each polynomial in $\Q_1$ and $\Q_2$ has degree at most $k-1$. Under these notations, we give  the following lemmas, which  generalize \cite[Lemma~II.2]{GreavesS19}. Proofs of Lemmas~\ref{SBGM1} and \ref{SBGM2} are given in Appendix.

 %



%
%

\begin{lemma}\label{SBGM1}
  Suppose that the $(k-1)$-uniform set system $\mathcal{S}=\{S_1,S_2,\ldots,S_k\}$  is non-intersecting and separable at some $i\in[k-1]$. Then
  \begin{gather*}
    C(\P)=\begin{bmatrix}
      C(\P_1) & 0\\
      0 & C(\P_2)
    \end{bmatrix}
    \begin{bmatrix}
      C(\Q_1)\\
      C(\Q_2)
    \end{bmatrix}.
  \end{gather*}
\end{lemma}

\begin{lemma}\label{SBGM2}
  The determinant of $[C(\Q_1)\; C(\Q_2)]^T$ is nonzero. In particular  $\det([C(\Q_1)\;$ $ C(\Q_2)]^T)=\prod_{u\in A,v\in B}(a_u-a_v)$.
\end{lemma}

 By Lemmas~\ref{SBGM1} and \ref{SBGM2}, we have
\[\det(C(\P))=\det(C(\P_1))\det(C(\P_2))\prod_{u\in A,v\in B}(a_{u}-a_{v}),\]
when $\S$ is non-intersecting and separable. To make sure that $\det(C(\P))\neq 0$, we need both $\det(C(\P_1))$ and $\det(C(\P_2))$ are nonzero, which are the coefficient matrices of polynomials defined by the residual set systems of $\mathcal{S}$. This motivates us to define a binary tree from a set system $\S$ in the following way. Let the $(k-1)$-uniform system $\mathcal{S}=\{S_1,S_2,\ldots S_k\}$ be the root node. If $\mathcal{S}$ is not non-intersecting and separable, then stop. Otherwise, let the two residual set systems $\A$ and $\B$ of $\mathcal{S}$  be the left and right children of $\mathcal{S}$. Then consider $\A$ and $\B$, say for example $\A$. If $\A$ is not non-intersecting and separable or $\A=\{\emptyset\}$, then stop. Otherwise, we can extend $\A$ by its two residual set systems. Keep doing this until  we can not extend any more. If the resulting binary tree has all leaf nodes being $\{\emptyset\}$, then we say that it is a {\it good} binary tree. Note that the binary tree constructed from $\S$ may not be unique. We say $\S$ is \emph{good} if it can produce at least one good binary tree. For convenience, we use the sequence of indices separating the nodes from top to bottom to indicate a specific binary tree, where indices for different layers are separated by a semi-colon. To be clear about the structure of $\S$, in the sequence, we use the original indices from $\S$ instead of indices from its decedents. See Example~\ref{eg1} and Fig.~\ref{picture}.


\begin{example}\label{eg1}
  Let $S_1=\{5,6,7,8\},S_2=\{1,6,7,8\},S_3=\{1,2,7,8\}, S_4=\{1,2,3,4\}, S_5=\{2,3,4,5\}.$ Then $\mathcal{S}=\{S_1$,$S_2$,
  $S_3$,$S_4$,$S_5\}$ is a $4$-uniform set system over $[8]$. We show that $\S$ is good and corresponds to a good binary tree.

   In fact, $\mathcal{S}$ is non-intersecting and separable at $3$ since $|S_1\cap S_2\cap S_3|=|\{7,8\}|=2$ and $|S_4\cap S_5|=|\{2,3,4\}|=3$. Then we have two residual set systems $\A=\{S_1^\prime,S_2^\prime,S_3^\prime\}$ and $\B=\{S_4^\prime,S_5^\prime\}$ with index $3$, where $S_1'=\{5,6\}, S_2'=\{1,6\}, S_3'=\{1,2\}, S_4'=\{1\}$ and $S_5'=\{5\}$. Both $\A$ and $\B$ are non-intersecting and separable since $S_1^\prime\cap S_2^\prime=\{6\}$. Let $S_1''=\{5\},S_2''=\{1\}, \mathcal{C}=\{S_1'',S_2''\}$. Then the corresponding binary tree in Fig.~\ref{picture} is denoted by $(3;2,4;1)$, which is good since all leaf nodes are $\{\emptyset\}$.
    \begin{figure}[!htbp]
\centering
\includegraphics[scale=0.35]{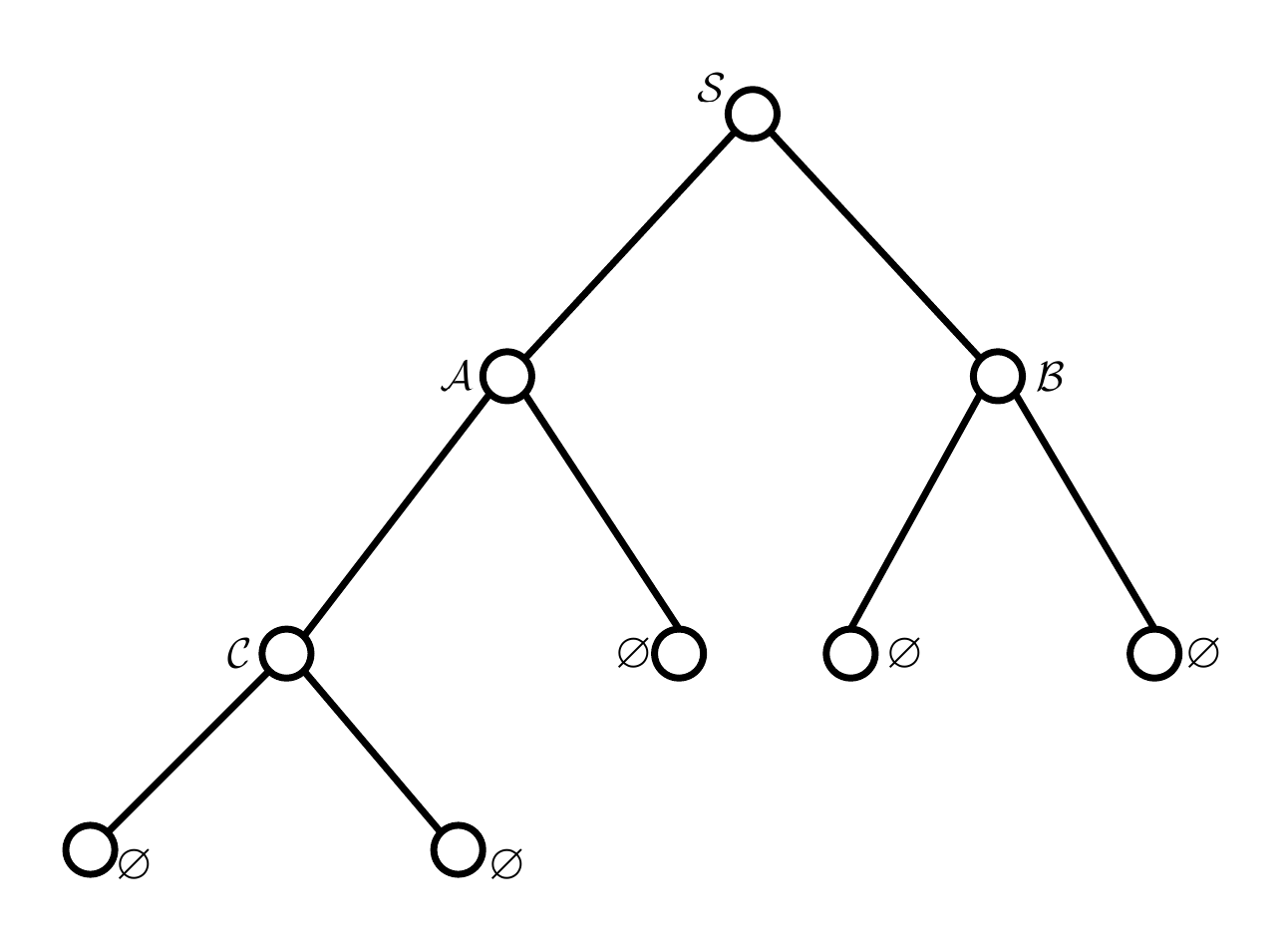}
\caption{The binary tree $(3;2,4;1)$ in Example~\ref{eg1}.}\label{picture}
\end{figure}

\end{example}


Note that if a set system has only one member which is an empty set, that is $\{\emptyset\}$, then there is only one constant polynomial defined by it. Hence the corresponding coefficient matrix is a $1\times 1$ matrix with entry $1$, i.e., $C(\{\emptyset\})=(1)$, which has determinant $1$.
Combining Lemmas~\ref{SBGM1} and \ref{SBGM2}, and all analysis before together, we know that if the $(k-1)$-uniform system $\mathcal{S}=\{S_1,\ldots,S_k\}$ is good, then picking any $n$ distinct elements $a_1,\ldots,a_n\in \bbF_q$ to define $P_{S_i}$, will result in an invertible coefficient matrix $C(\P)$, and consequently a generator matrix $G$ of an $[n,k]_q$ RS code whose complementary support matrix  is $M_{\S}$. We summarize this result into the following theorem.


\begin{theorem}\label{thm_MDS}
  Let $\mathcal{S}=\{S_1,\ldots,S_k\}$ be a good $(k-1)$-uniform set system over  $[n]$.  Then for any finite field $\mathbb{F}_q$ with $q\geq n$, there exists a sparse  $[n,k]_q$ RS code, whose generator matrix $G$ has the property that $G_{ij}=0$ if and only if $j\in S_i$, i.e., $M_{\mathcal{S}}$ is the complementary support matrix  of $G$.
\end{theorem}


\begin{example}\label{eg2}
  Let $n=8$, $\mathcal{S}=\{S_1,S_2,S_3,S_4,S_5\}$ be the set system given in Example~\ref{eg1}, then $\mathcal{S}$ satisfies the assumptions of Theorem~\ref{thm_MDS}. We will construct a $5\times 8$ matrix $G$ over $\mathbb{F}_q=GF(2^3)$ with the property that $G_{ij}=0$ if and only if $j\in S_i$. Let $\mathcal{P}=\{P_{S_1}(x),\cdots,P_{S_5}(x)\}$, where
  \[P_{S_1}(x)=(x-a_5)(x-a_6)(x-a_7)(x-a_8)\]
  \[P_{S_2}(x)=(x-a_1)(x-a_6)(x-a_7)(x-a_8)\]
  \[P_{S_3}(x)=(x-a_1)(x-a_2)(x-a_7)(x-a_8)\]
  \[P_{S_4}(x)=(x-a_1)(x-a_2)(x-a_3)(x-a_4)\]
  \[P_{S_5}(x)=(x-a_2)(x-a_3)(x-a_4)(x-a_5).\]
  Then the determinant of the coefficient matrix $C(\mathcal{P})$ is
    \[\det(C(\mathcal{P}))=\prod_{u\in\{7,8\},v\in\{2,3,4\}}(a_u-a_v)\cdot \prod_{v'\in\{1,2\}}(a_6-a_{v'}) \cdot (a_1-a_5) \cdot(a_5-a_1).\]

  The determinant of $C(\mathcal{P})$ will be nonzero in $\mathbb{F}_q$ if all $a_i$ are distinct elements of $\mathbb{F}_q$. Let $\zeta$ be a primitive element of $\mathbb{F}_q$ that satisfies $\zeta^3+\zeta+1=0$, and $a_1=0, a_i=\zeta^{i-2}$ for $i\in[2,8]$. Then $\det(C(P))=1\neq 0$, and the $5\times 8$ matrix $G$ over $\mathbb{F}_q$ is

  \[G=\left(\begin{array}{cccccccc}
      \zeta^4 & \zeta^5 & \zeta^6 & \zeta^2 & 0 & 0 & 0 & 0 \\
      0 & \zeta^4 & 1 & \zeta^6 & \zeta & 0 & 0 & 0 \\
      0 & 0 & \zeta & \zeta^4 & \zeta^3 & \zeta^5 & 0 & 0 \\
      0 & 0 & 0 & 0 & \zeta^2 & \zeta^5 & \zeta^4 & \zeta^6 \\
      \zeta^6 & 0 & 0 & 0 & 0 & 1 & \zeta & \zeta^4 \\
      \end{array}
      \right)
  \]

  By Theorem~\ref{thm_MDS}, there exists an $[8,5]$ RS code whose generator matrix is $G$. Furthermore, $G$ is also sparse and balanced.

\end{example}
A good $(k-1)$-uniform set system $\S$ satisfying the conditions of Theorem~\ref{thm_MDS} simply exists, for example,  let $S_i=[i,i+k-2] \mm n$. That is, sparse RS codes always exist. However, this can not give us a balanced code. To get a sparse and balanced RS code, we need to find a  good $(k-1)$-uniform set system, which simultaneously has almost the same element occurrences.  For convenience, we say a set system $\S$ or the matrix $M_{\S}$ is sparse and balanced if the matrix $G$ is sparse and balanced, and $M_{\S}$ is good if $\S$ is good.

\section{Constructions of Sparse and Balanced MDS codes}\label{conca}

In this section, we prove our main result Theorem~\ref{thm_main}, which gives the existence of MDS codes with an SBGM over a small field $\mathbb{F}_q$ with $q\geq n$.  For the reader's convenience, we restate Theorem~\ref{thm_main} here.


\begin{theorem*}[Theorem~\ref{thm_main}]
 For any integer $k\geq 3$, let $n\leq 2k$ if $k$ is even, and $n\leq 2k-1$ if $k$ is odd.
  For any finite field $\mathbb{F}_q$ with $q \geq n$, there exists an $[n,k]_q$ MDS code, whose generator matrix is sparse and balanced.
\end{theorem*}

By Theorem~\ref{thm_MDS}, we need to construct a good and balanced $(k-1)$-uniform set system $\S$. Equivalently, we need to construct a $k\times n$ binary matrix $M=(m_{i,j})=M_{\S}$ satisfying the following properties:

\begin{itemize}
  \item[($P_1$)] Sparse condition: each row of $M$ has $k-1$ zeros;
  \item[($P_2$)] Balanced condition: there are $\mu$ columns of $M$ containing $\lceil\frac{k(k-1)}{n}\rceil$ zeros and the rest containing $\lfloor\frac{k(k-1)}{n}\rfloor$ zeros, where $\mu\triangleq k(k-1)\mm n$;
  \item[($P_3$)] Good condition: for $i\in[k]$, let $S_i=\{j: m_{i,j}=0\}\subseteq[n]$, then the set system $\mathcal{S}=\{S_1,\ldots,S_k\}$ is good.
\end{itemize}

 In the sequential of this paper, we do not distinguish $\S$ and $M_{\S}$. Note that $S_i=[i,i+k-2] \mm n$, $i\in[k]$ gives a matrix $M_{\S}$ satisfying both ($P_1$) and ($P_3$). It is also easy to construct a matrix satisfying both ($P_1$) and ($P_2$). However, it is not easy to construct a matrix satisfying all of the three properties. We first give  necessary conditions in Theorem~\ref{thm_main} for the restriction of $n$ and $k$ in a matrix satisfying ($P_1$)-($P_3$).

\begin{lemma}\label{ivthm1}
  Given two positive integers $n\geq k\geq 3$, if there exists a sparse, good and balanced $k\times n$ binary matrix, then $n\leq 2k$. Further if $n=2k$, then $k$ must be even.
\end{lemma}

\begin{proof}
  Suppose $\S=\{S_1,\ldots,S_k\}$  and $M_{\S}$ is sparse, good and balanced. By the good condition, there exist $i\in [k-1]$, such that  $|S_1\cap S_2\cap\cdots\cap S_i|=k-i$ and $|S_{i+1}\cap S_{i+2}\cap\cdots\cap S_k|=i$. Suppose  that $n\geq 2k$. Then by the balanced condition, both $i$ and $k-i$ are at most $\lceil\frac{k(k-1)}{n}\rceil\leq\lceil\frac{k(k-1)}{2k}\rceil=\lceil\frac{k-1}{2}\rceil$. If  $i<\lceil\frac{k(k-1)}{n}\rceil\leq\lceil\frac{k-1}{2}\rceil$, then
  \begin{align*}
    k-i&\geq k-\left\lceil\frac{k-1}{2}\right\rceil+1=k-1-\left\lceil\frac{k-1}{2}\right\rceil+2\\
    &=\left\lfloor\frac{k-1}{2}\right\rfloor+2\geq \left\lceil\frac{k-1}{2}\right\rceil+1.
  \end{align*}
 This contradicts to $k-i\leq \lceil\frac{k-1}{2}\rceil$. The case $k-i<\lceil\frac{k(k-1)}{n}\rceil$ is similarly not possible.  Hence $i=k-i= \frac{k}{2}= \lceil\frac{k(k-1)}{n}\rceil.$ Hence, $k$ must be even. Since $\S$ is separable at $i=k/2$, then
there are at least $k$ columns in $M_{\mathcal{S}}$ each containing at least $\frac{k}{2}$ zeros. This  excludes the cases when $n=2k+1,2k+2$ for all $k\geq 4$, or $n=2k+3$ and for $k=4$, since for these cases, the value $\mu=k(k-1)\mm n$ is strictly less than $k$.

 Next, we exclude all other cases due to the equation  $\frac{k}{2}= \lceil\frac{k(k-1)}{n}\rceil$.  Assume that $n=2k+t$ with $t\geq 4$, or $t=3$ and $k\geq 6$. Then
 \[\left\lceil\frac{k(k-1)}{n}\right\rceil=\left\lceil\frac{k(k-1)}{2k+t}\right\rceil=\left\lceil\frac{k}{2}-\frac{(t+2)k}{4k+2t}\right\rceil\leq \frac{k}{2}-1,\] thus a contradiction.

So we conclude that $n=2k$ and $k$ is even by assumption, or $n<2k$.
\end{proof}

By Lemma~\ref{ivthm1}, we only need to consider $n<2k$ or $n=2k$ and $k$ is even. When $n=2k$ and $k$ is even,  we have $\mu=k$, that is, there are exactly $k$ columns each containing $\frac{k}{2}$ zeros and $k$ columns each containing $\frac{k}{2}-1$ zeros in $M_{\mathcal{S}}$.
A valid matrix can be constructed by cyclically shifting the vectors $(\underbrace{0,\ldots,0}_{k-1},1,\ldots,1)$ and $(\underbrace{1,\ldots,1}_k,\underbrace{0,\ldots,0}_{k-1},1)$ each for $\frac{k}{2}$ times. Formally, the construction is given below.

\begin{construction}\label{con_k=d-1}
  Let $n=2k$ for any even $k\geq 4$. For $i\in[1,\frac{k}{2}]$, let $S_i=[i,i+k-2]$; for $i\in[\frac{k}{2}+1,k]$, let $S_i=[\frac{k}{2}+i,\frac{3k}{2}+i-2]\mm n$. Then the set system $\mathcal{S}=\{S_1,\ldots,S_k\}$ is good due to the good binary tree $(k/2;1,k-1;2,k-2;\ldots;k/2-1,k/2+1)$. Further, the $j$th column in $M_{\mathcal{S}}$ has $\frac{k}{2}-1$ zeros if $j\in[1,\frac{k}{2}-1]\cup [k,\frac{3k}{2}-1]\cup\{2k\}$, and $\frac{k}{2}$ zeros if $j\in [\frac{k}{2},k-1]\cup[\frac{3k}{2},2k-1]$. Hence $M_{\mathcal{S}}$ is sparse and balanced.
\end{construction}


From now on, we assume that $n<2k$.  We next prove a key ingredient in our algorithms.

\subsection{Key Operations in Algorithms}

%
%
%
%
%

Notice that in Construction~\ref{con_k=d-1}, the matrix $M_{\mathcal{S}}$ is obtained by cyclically shifting two vectors. When $n<2k$, this method fails since the number of zeros is always $k(k-1)$, but $n$ becomes smaller, which may lead to insufficient zeros in the middle columns of $M_{\mathcal{S}}$.

A circulant matrix from a binary vector is good since it trivially satisfies the sparse and good conditions.  However, it is not balanced in general. Next, we show that if $M_{\mathcal{S}}$ is a circulant binary matrix, then we can  adjust the positions of zeros restricted in the same row, so that the new $M_\mathcal{S}$ is still sparse and good, but becomes more balanced. See the following example.


\begin{example}\label{eg3.1}
  Suppose $k=6$ and $n=10$. Let $M_{\mathcal{S}}$ be the following matrix which is obtained by cyclically shifting the vector $(0, 0, 0, 0, 0,1,1,1,1,1)$. Note that $\S$ can produce a good binary tree $(3;1,5;2,4)$. We partition $M_{\mathcal{S}}$ into several blocks by drawing lines  at the $(i-1)$th  and the $(i+1)$th rows,  the $(i-1)$th  and the $(i+k-1)$th columns for $i=3$.
  \[M_{\mathcal{S}}=\left(\begin{array}{cc|cccccc|cc}
      0 & 0 & 0 & 0 & 0 & \textcolor[rgb]{0.98,0.00,0.00}{1} & 1 & 1 & 1 & 1 \\
      1 & \textcolor[rgb]{0.00,0.00,1.00}{0} & 0 & 0 & 0 & 0 & \textcolor[rgb]{0.98,0.00,0.00}{1} & 1 & 1 & 1 \\
      \hline
      1 & 1 & 0 & 0 & 0 & 0 & \textcolor[rgb]{0.00,0.00,1.00}{0} & 1 & 1 & 1 \\
      1 & 1 & 1 & \textcolor[rgb]{1.00,1.00,0.00}{0} & 0 & 0 & 0 & 0 & 1 & 1 \\
      \hline
      1 & 1 & 1 & \textcolor[rgb]{0.98,0.00,0.00}{1} & 0 & 0 & 0 & 0 & \textcolor[rgb]{1.00,1.00,0.00}{0} & 1 \\
      1 & 1 & 1 & 1 & \textcolor[rgb]{0.98,0.00,0.00}{1} & 0 & 0 & 0 & 0 & 0 \\
      \end{array}
      \right)
  \]

  All our exchanges will be restricted in the same row. We observe that the following exchanges do not destroy the good condition. The zeros in the upper left and lower right corners can be exchanged with any ones in the row where they are, except for the red ones.  For the blue and yellow zeros in the center block, we can exchange them with any one in the same row, but the two blue and the two yellow  zeros cannot be in the same column, respectively. Otherwise, they produce repeated rows. For all other zeros, let them stay where they are. For example, we can update $M_{\mathcal{S}}$ to the following $M_{\mathcal{S'}}$, where the binary tree $(3;1,5;2,4)$ from $\S'$ is still  good.
  \[M_{\mathcal{S'}}=\left(\begin{array}{cc|cccccc|cc}
      1 & 0 & 0 & 0 & 0 & 1 & 1 & 1 & 1 & 0 \\
      1 & 1 & 0 & 0 & 0 & 0 & 1 & 1 & 0 & 1 \\
      \hline
      1 & 1 & 0 & 0 & 0 & 0 & 1 & 1 & 1 & 0 \\
      1 & 0 & 1 & 1 & 0 & 0 & 0 & 0 & 1 & 1 \\
      \hline
      0 & 1 & 1 & 1 & 0 & 0 & 0 & 0 & 1 & 1 \\
      1 & 0 & 0 & 1 & 1 & 0 & 0 & 0 & 1 & 1 \\
      \end{array}
      \right)
  \]

\end{example}

We extend Example~\ref{eg3.1} to a more general case in the following lemma.

\begin{lemma}\label{thm_base}(Key Operations)
  Given positive integers  $n\geq k$ and $\alpha\in[n]$. Let $\mathcal{S}=\{S_1,\ldots,S_k\}$ be a $(k-1)$-uniform set system which corresponds to a good binary tree with the first layer index $\beta$. Suppose $S_j=[\alpha+j-1,\alpha+j+k-3] \mm n$, $j\in [\beta-1]$. For any $s\in[\beta]$, define a new set system $\mathcal{S}'$ with each set of them $\mm n$ as follows:
   \begin{itemize}
   \item[(1)] for any $t\in[s-1]$, $S_{t}'=[\alpha+s-1,\alpha+k+t-3]\cup\{a_{t,1},\ldots,a_{t,s-t}\}$ with $\alpha+k+t-2\not\in\{a_{t,1},\ldots,a_{t,s-t}\}$;
           \item[(2)] for $t\in[s+2,\beta]$, $S_{t}'=[\alpha+t-1,\alpha+k+s-2]\cup\{a_{t,1},\ldots,a_{t,t-s-1}\}$ with $\alpha+t-2\not\in\{a_{t,1},\ldots,a_{t,t-s-1}\}$;
       \item[(3)] $S_{s}'=[\alpha+s-1,\alpha+k+s-3]\cup\{a_{s,1}\}$ with $a_{s,1}\neq a_{s-1,1}$ and $S_{s+1}'=[\alpha+s+1,\alpha+k+s-2]\cup\{a_{s+1,1}\}$ with $a_{s+1,1}\neq a_{s+2,1}$.
\end{itemize}
               Then the new set system $\mathcal{S}'=\{S_1',\ldots,S_{\beta}',S_{\beta+1},\ldots,S_k\}$ is still good.

\end{lemma}

\begin{proof} Let $\bar{\S}=\{S_1,\ldots,S_{\beta}\}$, consider the matrix $M_{\bar{\mathcal{S}}}$. The left picture in Fig.~\ref{picture2} draws the positions of zeros in $M_{\bar{\mathcal{S}}}$, which are enclosed by the solid lines.  Each row of $M_{\bar{\mathcal{S}}}$ represents a set $S_i$, $i\in[\beta]$. In this picture, we assume that $\alpha=1$ and $k+\beta-2<n$ for simplicity, and the proof is still true if we remove these conditions. Let $\bar{\S'}=\{S'_1,\ldots,S'_{\beta}\}$. Then in the new matrix $M_{\bar{\mathcal{S}'}}$, we only exchange zeros in the red area with any ones in its own row, except for the ones identified by red stars. Further,  the exchange of the two red zeros are restricted by (3), so that no repeat rows are produced.


Since $S'_1\cap \cdots\cap S'_{\beta}=[\alpha+\beta-1,\alpha+k-2]$, which is of size $k-\beta$, then $\S'$ is still separable at $\beta$.
Let $\mathcal{A'}=\{A_1',\ldots,A_\beta'\}$ be a $(\beta-1)$-uniform system with $A_i'=S_i'\setminus [\alpha+\beta-1,\alpha+k-2]$ for each $i\in[\beta]$.  We only need to show that $\A'$ is good. Since $|A_1'\cap\cdots\cap A_s'|=|[\alpha+s-1,\alpha+\beta-2]|=\beta-s$ and $|A_{s+1}'\cap\cdots\cap A_\beta'|=|[\alpha+k-1,\alpha+k+s-2]|=s$, $\mathcal{A'}$ is non-intersecting and separable at $s$, and has two residual set systems $\mathcal{C'}=\{C_1',\ldots,C_s'\}$ with $C_i'=A_i'\setminus [\alpha+s-1,\alpha+\beta-2]$ for $i\in[s]$, and $\mathcal{D'}=\{D_1',\ldots,D_{\beta-s}'\}$ with $D_j'=A_{s+j}'\setminus [\alpha+k-1,\alpha+k+s-2]$ for $j\in[\beta-s]$. Hence $\mathcal{C'}$ and $ \mathcal{D'}$ are two descendants of $\mathcal{A'}$. It is easy to check that both  $\mathcal{C'}$ and $ \mathcal{D'}$ correspond to a good binary tree, and so does $\A'$. In fact, the binary tree $(s;1,\beta-1;2,\beta-2;\ldots)$ from $\A'$ is  good. See the right figure in Fig.~\ref{picture2}, where all red nodes are $\{\emptyset\}$.
\end{proof}



\begin{figure*}[!htbp]
\centering
\includegraphics[scale=0.9]{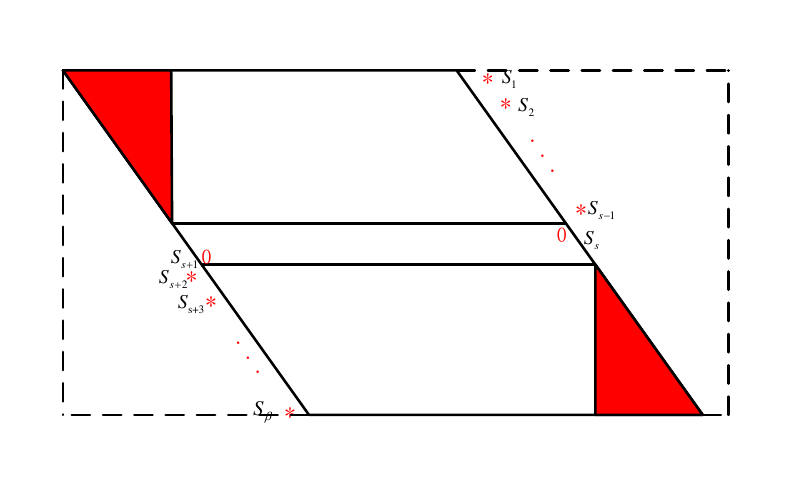}
\includegraphics[scale=0.4]{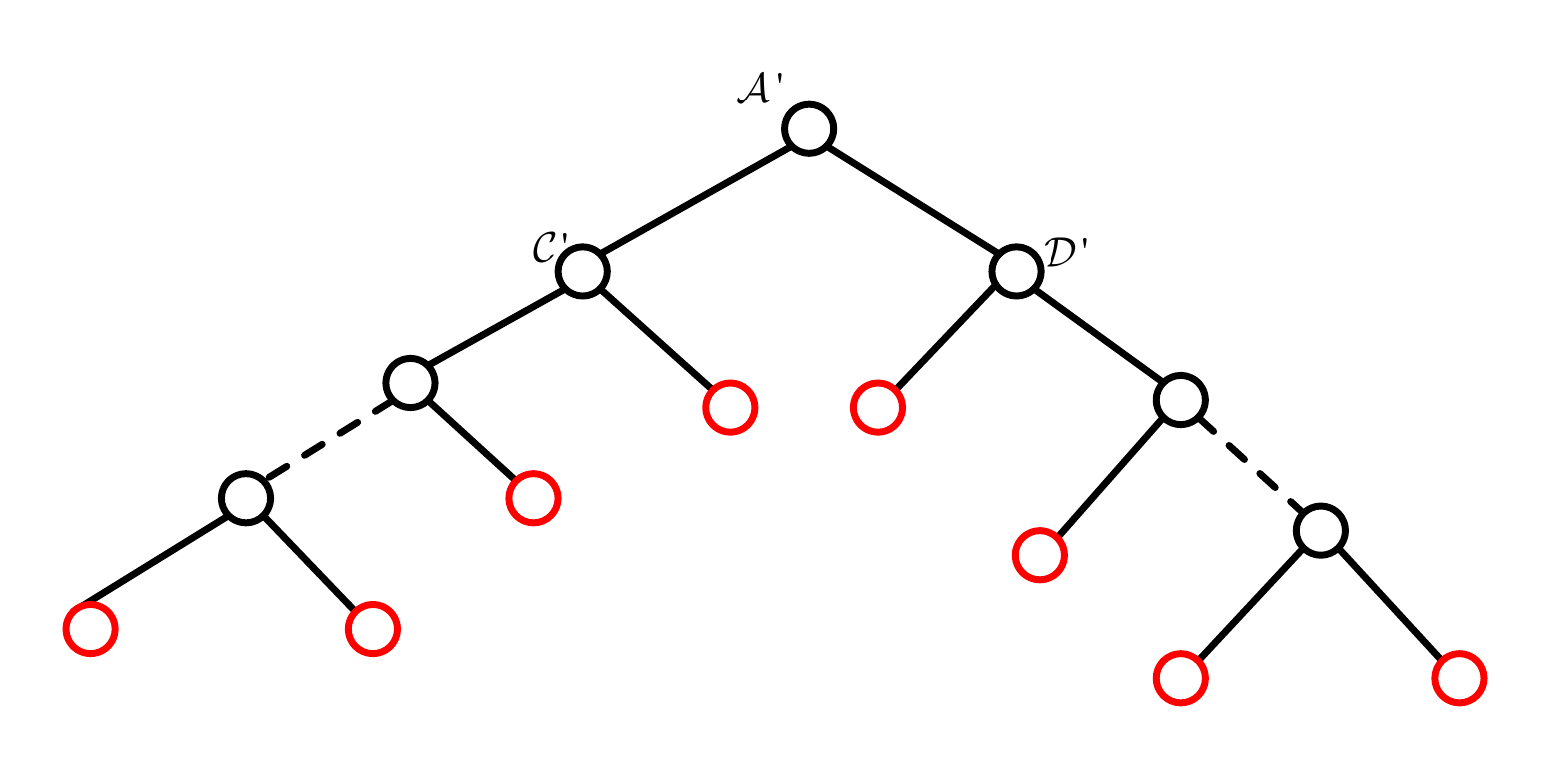}
\caption{The left picture draws the positions of zeros in $M_{\bar{\mathcal{S}}}$, while  the right picture is the  binary tree $(s;1,\beta-1;2,\beta-2;\ldots)$ from $\A'$. }\label{picture2}
\end{figure*}

 In Lemma~\ref{thm_base}, $\S$ is separable at $\beta$. So we can update the last $k-\beta$ rows simultaneously, and the resulting $\mathcal{S}'=\{S_1',\ldots,S_{\beta}',S_{\beta+1}',\ldots,S_k'\}$ is still good. Example~\ref{eg3.1} is a case of Lemma~\ref{thm_base}  by picking $\alpha=1,\beta=6$ and $s=3$.

Lemma~\ref{thm_base} tells us that, starting from a circulant block of $M_{\S}$, which satisfies the sparse and  good condition, we can modify it to a more balanced block keeping the sparse and  good property. In the next subsections, we will apply Lemma~\ref{thm_base} repeatedly in our algorithms to output a sparse, good and balanced matrix, which  could be used as the complementary support matrix of a generator matrix for an RS code.

\subsection{Constructions of $M_{\mathcal{S}}$ with $n<2k$}
In this section, we show the existence of a sparse, good and balanced $k\times n$ binary matrix for all $n<2k$. Let $n=2k-t$ with $t\in[k]$. For convenience, we always assume that $k$ is even. The constructions for odd $k$  are similar and can be provided upon requests.  For $t=1,2$, we construct the desired matrix explicitly. For $t\in[3,k]$, we show the existence of such a matrix by several algorithms of applying Lemma~\ref{thm_base}.


 Let $a\triangleq\lceil\frac{k(k-1)}{n}\rceil$ and $b\triangleq\lfloor\frac{k(k-1)}{n}\rfloor$ be the required numbers of zeros in each column. Remember $\mu\triangleq k(k-1)\mm n$ is the  required number of columns each containing $a$ zeros.

\begin{construction} When $t=1$, then $n=2k-1$, $\mu= \frac{3k}{2}-1$ and $a=\frac{k}{2}$.
   Construct  $\mathcal{S}=\{S_1,\ldots,S_k\}$ as follows: for $i\in[1,\frac{k}{2}]$, $S_i=[i,i+k-2]$; for $i\in[\frac{k}{2}+1,k]$, $S_i=[1,i-\frac{k}{2}-2]\cup [\frac{k}{2}+i-1,2k-1]$.  The good binary tree from $\S$ is $(\frac{k}{2};1,k-1;2,k-2;\ldots;\frac{k}{2}-1,\frac{k}{2}+1)$.
\end{construction}



\begin{construction}When $t=2$, then $n=2k-2$,  $\mu= n$ and $a=\frac{k}{2}$.
  Construct  $\mathcal{S}=\{S_1,\ldots,S_k\}$ as follows:  for $i\in[1,\frac{k}{2}]$, $S_i=[i,i+k-2]$; for $i\in[\frac{k}{2}+1,k]$, $S_i=[1,i-\frac{k}{2}-1]\cup [\frac{k}{2}+i-1,2k-2]$. The good binary tree from $\S$ is $(\frac{k}{2};1,k-1;2,k-2;\ldots;\frac{k}{2}-1,\frac{k}{2}+1)$.
\end{construction}

%

%

When $t\in[3,k]$, write $t=4m+u$ with $u=0,1,2,3$. Then \[\frac{k(k-1)}{2k-t}=\frac{k(k-\frac{t}{2})+(\frac{t}{2}-1)k}{2(k-\frac{t}{2})}=\frac{k}{2}+\frac{(t-2)k}{4(k-\frac{t}{2})}=\frac{k}{2}+
\frac{t-2}{4}+\frac{t^2-2t}{8k-4t}.\]
When $3\leq t<\frac{1+\sqrt{8k}}{2}$, we have $\frac{t^2-2t}{8k-4t}<\frac{1}{4}$. Then $a=\lceil\frac{k(k-1)}{n}\rceil=\frac{k+r}{2}$ and $b=\lfloor\frac{k(k-1)}{n}\rfloor=\frac{k+r-2}{2}$, where the values of $r$ are depicted in Table~\ref{tr}.

\begin{table*}
\center
\caption{The distribution of $r$.}\label{tr}
\vspace{-0.3cm}
\[\begin{array}{c|c|c|c|c}
\hline
 & u=0 & u=1 & u=2 & u=3  \\
\hline
r & \frac{t}{2} & \frac{t-1}{2} & \frac{t+2}{2} & \frac{t+1}{2} \\
\hline

\end{array}\]
\end{table*}

Now we assume that $3\leq t<\frac{1+\sqrt{8k}}{2}$, and take $u=1$ as an example to illustrate our algorithm. All the other three cases are similar. In this case, $t=4m+1$ with $m<\frac{2\sqrt{2 k}-1}{8}$, $a=\frac{k+2m}{2},b=\frac{k+2m-2}{2}$ and $\mu=\frac{3k}{2}+4m^2-3m-1$.

  Start from an initial matrix  $M_\mathcal{S}=[M_1,M_2]^\top$, where $M_1$ and $M_2$ are both  circulant $\frac{k}{2}\times n$ matrices. The first row of $M_1$ corresponds to the set $[k-1]$, and the first row of $M_2$ corresponds to the set $[k-3m+1, 2k-3m-1]\mm n=[k-3m+1,n]\cup[1,m]$. See the left picture in Fig.~\ref{rudiment8} about the zero positions of $M_\mathcal{S}$.  All zeros are in the enclosed area by solid lines.  It is easy to check that $M_\mathcal{S}$ is sparse and good, but not balanced. We will apply the key operations in Lemma~\ref{thm_base} to make $M_\mathcal{S}$ a balanced matrix. The construction is given in Algorithm~\ref{alg:u2=1}, and we illustrate it in the right picture of Fig.~\ref{rudiment8}.

  \begin{figure*}[!htbp]
\centering
\includegraphics[scale=0.4]{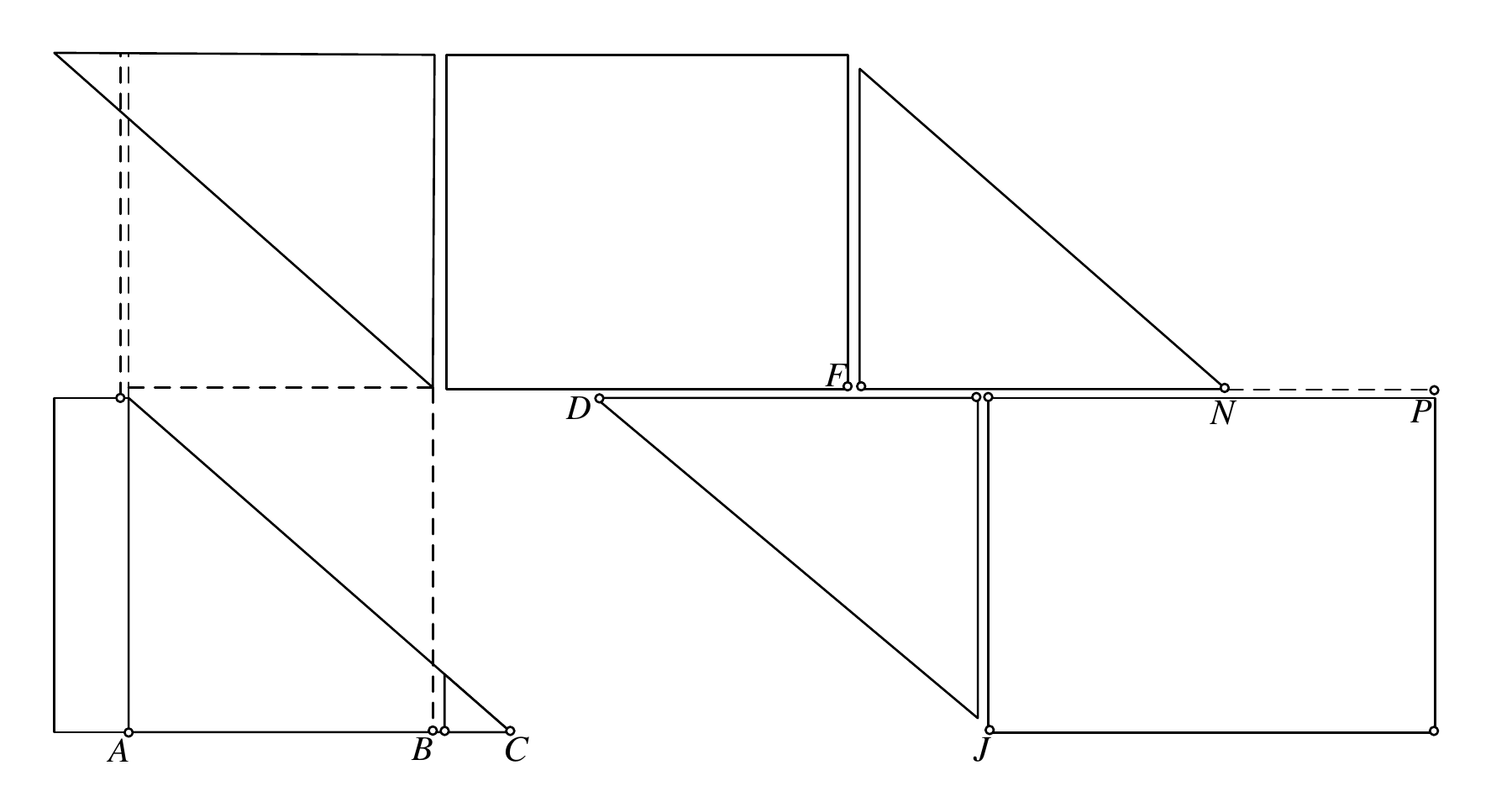}
\includegraphics[scale=0.4]{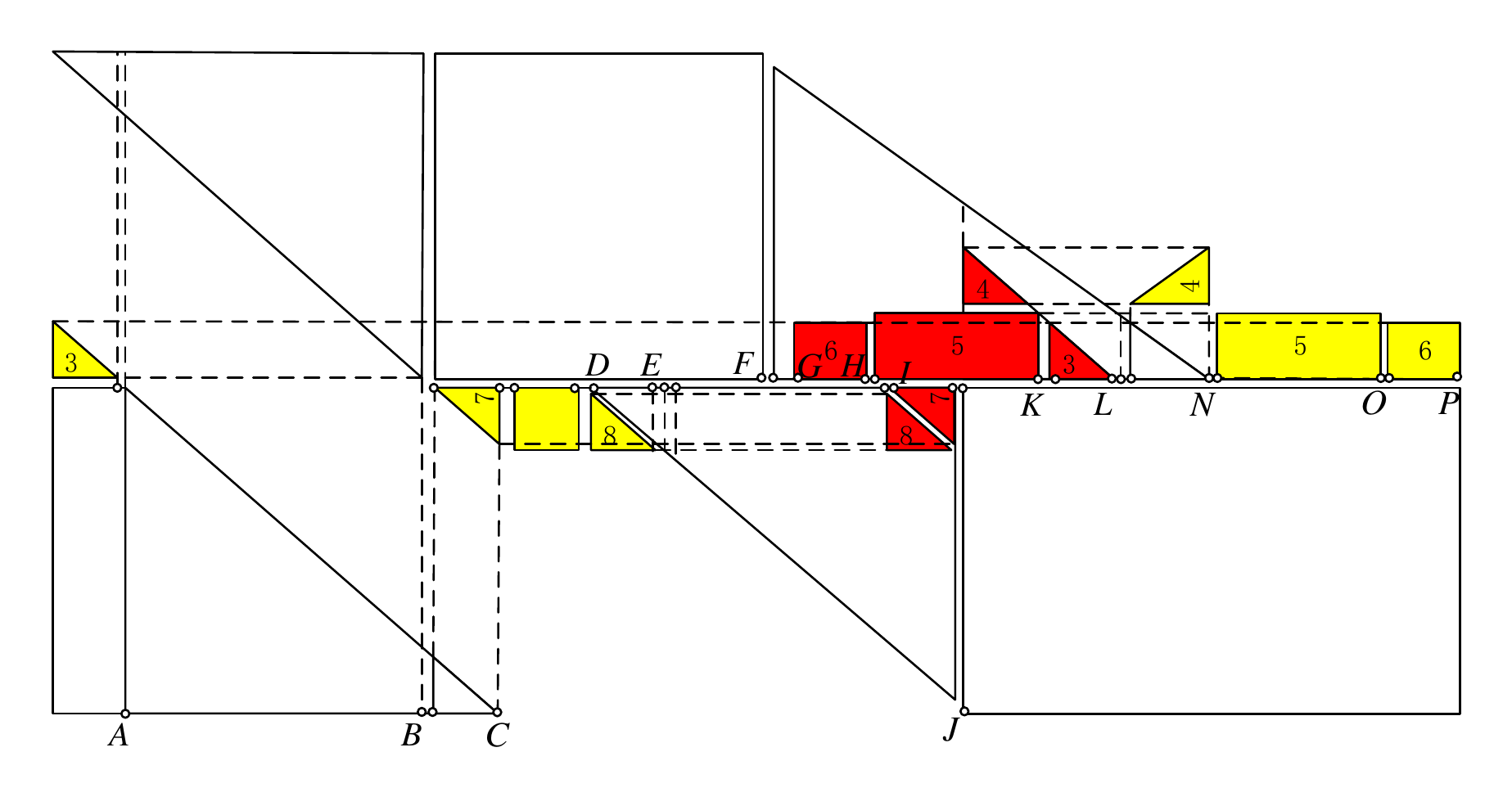}
\caption{ The left picture is the initial matrix $M_\mathcal{S}$.  The right picture is an illustration of Algorithm~\ref{alg:u2=1}.  A pair of blocks with the same index indicates the beginning (red) and ending (yellow) of a step with the same index in the algorithm. For example, in Step~\ref{alg1:u83} of Algorithm~\ref{alg:u2=1}, zeros in the red triangle with index $3$ move to the place of the yellow one.  The column coordinates of the points in these two figures are the same and listed as follows: $A=m$, $B=\frac{k}{2}$, $C=\frac{k}{2}+m-1$, $D=k-3m+1$, $E=k-2m-1$, $F=k-1$, $G=k+2m-1$, $H=\frac{3k}{2}-4m^2+m-2$, $I=\frac{3k}{2}-4m+1$, $J=\frac{3k}{2}-3m$, $K=\frac{3k}{2}-2m-1$, $L=\frac{3k}{2}-m-2$, $N=\frac{3k}{2}-2$, $O=\frac{3k}{2}+4m^2-3m-1$, $P=n$. Furthermore, the coordinates of two adjacent points differ by one, so do all the following figures.}\label{rudiment8}
\end{figure*}

%



\begin{algorithm}[!htbp]
  \caption{ Construction of $M_\mathcal{S}$ with $t=4m+1$ and $k$ is even.}
  \label{alg:u2=1}
  \begin{algorithmic}[1]
    \Require
      Integers $n,k,t=4m+1$ with $n\geq k$ and $n=2k-t$;
    \Ensure
      A sparse, good and balanced binary matrix $M_{\mathcal{S}}$.
    \State Construct the initial matrix $M_\mathcal{S}=(m_{i,j})$ as follows: for $i\in[1,\frac{k}{2}]$ and $j\in[i,i+k-2]$, $m_{i,j}=0$; for $i\in[\frac{k}{2}+1,k],j\in[1,m-1+i-\frac{k}{2}]\cup[k-3m+i-\frac{k}{2},n]$, $m_{i,j}=0$; for all the rest positions, $m_{i,j}=1$.
    \label{alg1:u81}
    \State Change the positions of zeros.
    \label{alg1:u82}
    \State For $i\in[\frac{k}{2}-m+2,\frac{k}{2}],j\in[1,i-(\frac{k}{2}-m+1)]$ and $j'\in[\frac{3k}{2}-2m,\frac{3k}{2}-2m+i-(\frac{k}{2}-m+2)]$, $m_{i,j}=0$ and $m_{i,j'}=1$.
    \label{alg1:u83}
    \State For $i\in[\frac{k}{2}-2m+2,\frac{k}{2}-m],j\in[\frac{3k}{2}-3m,\frac{3k}{2}-3m+i-(\frac{k}{2}-2m+2)]$ and $j'\in[\frac{3k}{2}-2-(i-(\frac{k}{2}-2m+2)),\frac{3k}{2}-2]$, $m_{i,j}=1$ and $m_{i,j'}=0$.
    \label{alg1:u84}
    \State For $i\in[\frac{k}{2}-m+1,\frac{k}{2}],j\in[\frac{3k}{2}-4m^2+m-1,\frac{3k}{2}-2m-1]$ and $j'\in[\frac{3k}{2}-1,\frac{3k}{2}+4m^2-3m-1]$, $m_{i,j}=1$ and $m_{i,j'}=0$.
    \label{alg1:u85}
    \State For $i\in[\frac{k}{2}-m+2,\frac{k}{2}],j\in[k+2m-1,\frac{3k}{2}-4m^2+m-2]$ and $j'\in[\frac{3k}{2}+4m^2-3m,n]$, $m_{i,j}=1$ and $m_{i,j'}=0$.
    \label{alg1:u86}
    \State For $i\in[\frac{k}{2}+1,\frac{k}{2}+m-1],j\in[i,\frac{k}{2}+m-1]$ and $j'\in[\frac{3k}{2}-4m+(i-\frac{k}{2}),\frac{3k}{2}-3m-1]$, $m_{i,j}=0$ and $m_{i,j'}=1$.
    \label{alg1:u87}
    \State For $i\in[\frac{k}{2}+2,\frac{k}{2}+m],j\in[k-3m+1,k-3m+(i-\frac{k}{2}-1)]$ and $j'\in[\frac{3k}{2}-4m,\frac{3k}{2}-4m-1+(i-\frac{k}{2}-1)]$, $m_{i,j}=0$ and $m_{i,j'}=1$.
    \label{alg1:u88}
    \State For any two columns $j\in[\frac{k}{2}+m,k-3m]$ and $j'\in[k-2m+1,\frac{3k}{2}-4m-1]$. Find a row $i\geq \frac{k}{2}+1$ satisfying $m_{i,j}=1$ and $m_{i,j'}=0$, then swap them: $m_{i,j}=0$ and $m_{i,j'}=1$.
    \label{alg1:u89}
    \State Repeat Step \ref{alg1:u89} until all columns from $[\frac{k}{2}+m,k-3m]\cup[k-2m+1,\frac{3k}{2}-4m-1]$ have $\frac{k+2m}{2}$ zeros.
    \label{alg1:u810}\\
    \Return $M_\mathcal{S}$;
  \end{algorithmic}
\end{algorithm}


Now we explain Algorithm~\ref{alg:u2=1} by following the notations in Fig.~\ref{rudiment8}. In the initial matrix $M_\mathcal{S}$, each column  with index in $[A,B]$ already has  exactly $a$ zeros.
After Steps~\ref{alg1:u83}-\ref{alg1:u86} of Algorithm~\ref{alg:u2=1}, the number of zeros of column $j\in[1,A-1]\cup[J,O]$ is $a$, and the number of zeros of column $j\in[O+1,P]$ is $b$. Notice that $|[O+1,P]|=\frac{k}{2}-4m^2-m=n-\mu$, so these are the all columns containing $b$ zeros, and we need to make all other columns to contain $a$ zeros. Thus, for each column $j$ from $[E+2,F]$, or $[F+1,G-1]$, or $[G,H]$, or $[H+1,J-1]$, we need to delete $j-(k-2m),2m-1,m,m-1$ zeros, respectively. We do not need to modify the $(E+1)$th column since it already has $a$ zeros. Steps~\ref{alg1:u87} and \ref{alg1:u88} further make the number of zeros in columns $[B+1,C]\cup[D,E]\cup[I-1,J-1]$ to $a$.


The Steps~\ref{alg1:u83}-\ref{alg1:u88} are explicit.  It is left to check the feasibility of Steps~\ref{alg1:u89} and \ref{alg1:u810}  in Algorithm~\ref{alg:u2=1}. There are $m$ zeros to be moved in for each column from $[C+1,D-1]$. For each column $j$ in $[E+2,F]$, or $[F+1,G-1]$, or $[G,H]$, or $[H+1,I-2]$, we need to move out $j-(k-2m),2m-1,m,m-1$ zeros from column $j$, respectively. Hence we only need to check whether the number of zeros we move in is the same as the number of zeros we move out. The number of zeros we need to move in is $(\frac{k}{2}-4m+1)m=\frac{k}{2}m-4m^2+m$, and the number of zeros we need to move out is
\[m(2m-1)+(2m-1)(2m-1)+(\frac{k}{2}-4m^2-m)m+(4m^2-4m+1-m)(m-1)=\frac{k}{2}m-4m^2+m.\]
Thus Steps~\ref{alg1:u89} and \ref{alg1:u810} in  Algorithm~\ref{alg:u2=1} are feasible, which finally gives us a sparse and balanced binary matrix. Since the initial matrix $M_\mathcal{S}$ is good, and all steps in Algorithm~\ref{alg:u2=1} satisfy the key operations of Lemma~\ref{thm_base}, the final new matrix $M_{\mathcal{S}}$ still corresponds to a good binary tree which is $(\frac{k}{2};\frac{k}{2}-1,\frac{k}{2}+1;\frac{k}{2}-2,\frac{k}{2}+2;\ldots;1,k-1)$.


In the sections to follow, we consider $t\geq \frac{1+\sqrt{8k}}{2}$, and  give three algorithms for constructing $M_{\mathcal{S}}$ for different ranges of $t$.  For convenience, we assume that $\mu\geq \frac{n}{2}$ and $k>176$, and the case when $\mu< \frac{n}{2}$ or $k\leq 176$ are similar. In fact, these three algorithms are generalizations of Algorithm~\ref{alg:u2=1}, but not explicit any more, since we do not know the exact values of $\mu$ and $a$.

\subsubsection{When $t\in[\lceil\frac{1+\sqrt{8k}}{2}\rceil,\frac{k}{2}+2]$}
\
\newline
\indent In this case, $n=2k-t\geq\frac{3k}{2}-2$ and $2k-a-1\geq k+a-t+1$. Start from an initial matrix  $M_{\mathcal{S}}=[M_1,M_2]^\top$, where $M_1$ and $M_2$ are both  circulant $\frac{k}{2}\times n$ matrices. The first row of $M_1$ corresponds to the set $[k-1]$, and the first row of $M_2$ corresponds to the set $[\frac{k}{2}+a+2-t, \frac{3k}{2}+a-t]\mm n=[\frac{k}{2}+a+2-t,n]\cup[1,a-\frac{k}{2}]$. See Fig.~\ref{algidea} about the zero positions of $M_{\mathcal{S}}$, where yellow areas are excluded.


%
%


  It is easy to check that $M_{\mathcal{S}}$ is sparse and good, but not balanced.
The columns with index in $[a-\frac{k}{2},\frac{k}{2}]$ each have $a$ zeros. Since  $\mu\geq \frac{n}{2}\geq\frac{3k}{4}-1$, we do not have to modify these columns. In order that each column from $[1,a-\frac{k}{2}-1]$ has $a$ zeros, we need some zero blocks as the yellow triangle labeled by $4$ in Fig.~\ref{algidea}, which can be moved from the red triangle labeled by $4$ in columns $[2k-a,\frac{3k}{2}-2]$. After this modification, any column in $[2k-a,n]$ has $\frac{k}{2}$ zeros. We need to move some zeros from the $(i,j)$th position with $i\in[2,\frac{k}{2}]$ and $j\in[k,2k-a-1]$ to columns $[2k-a,n]$, since each column from $[k,2k-a-1]$ has at least $a$ zeros.

\begin{figure}[!htbp]
\centering
\includegraphics[scale=0.6]{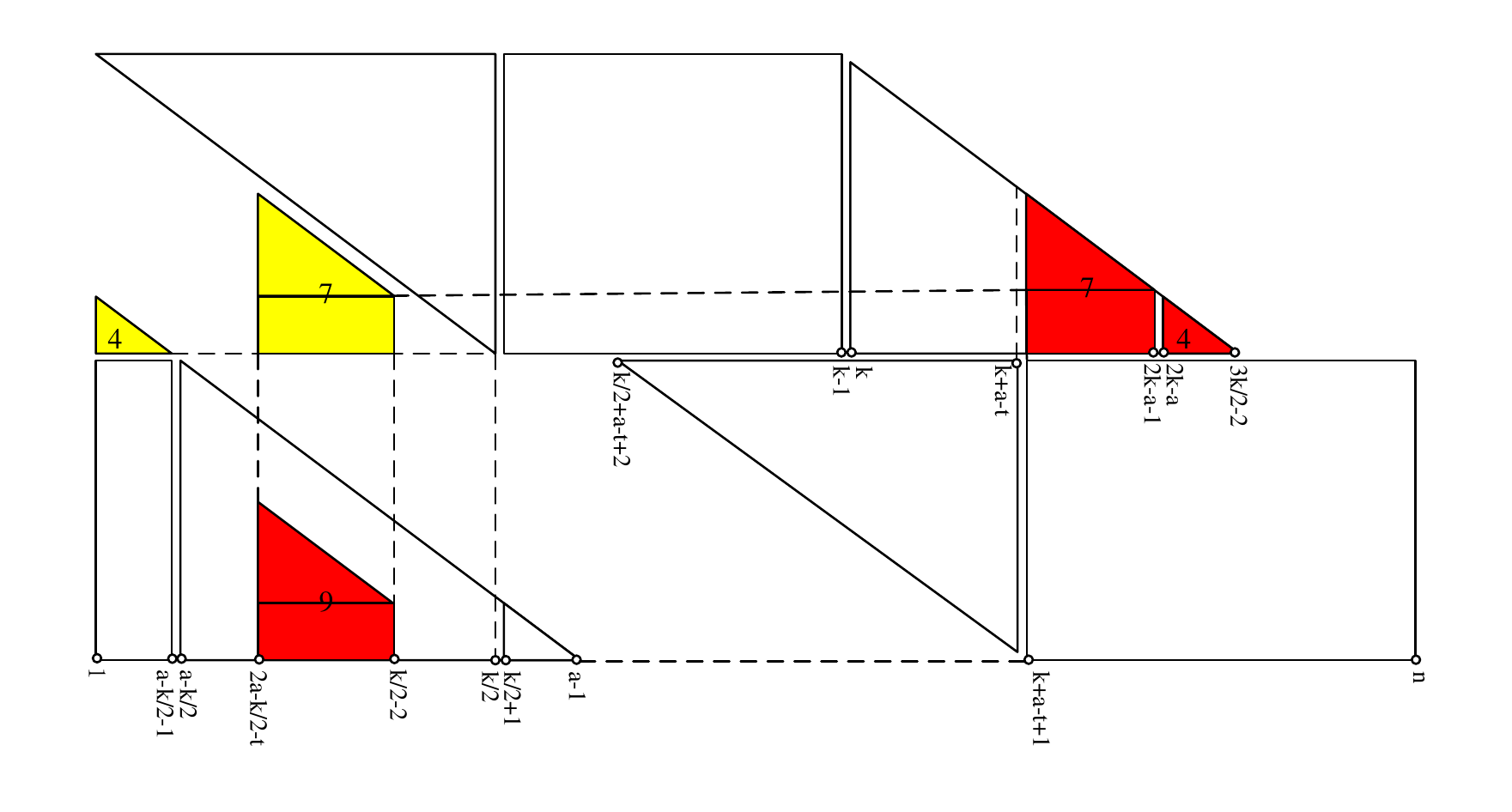}
\caption{The illustration of Algorithm~\ref{alg1:v}. The indices in the blocks are consistent with the indices of the Steps in algorithms, while the red  and yellow blocks correspond to the beginning and ending places of a movement.}\label{algidea}
\end{figure}

The most challenge of our approach is that after the above modifications, all columns in $[1,\frac{k}{2}]\cup[2k-a,n]$ has $a$ zeros, but there are still some columns in $[k+a-t+1,2k-a-1]$ having extra zeros; these extra zeros can not be moved to any other columns in row $[1,\frac{k}{2}]$ directly since it will increase extra zeros to columns in $[1,\frac{k}{2}]\cup[2k-a,n]$; and they are also can not be moved to any other columns in rows $[\frac{k}{2}+1,k]$, otherwise it may destroy the good condition of the bottom block since they are not key operations allowed in Lemma~\ref{thm_base}. 

To make the matrix balanced while preserving the good property of $M_{\mathcal{S}}$, we first move the extra zeros from columns $[k+a-t+1,2k-a-1]$ which are inside the red echelon labeled by $7$ in Fig.~\ref{algidea}, to columns $[2a-\frac{k}{2}-t,\frac{k}{2}-2]$ which are inside the yellow echelon labeled by $7$. Then the columns in $[2a-\frac{k}{2}-t,\frac{k}{2}-2]$ have extra zeros, so we move these extra zeros identified by red echelon labeled by $9$ in Fig.~\ref{algidea} to columns in $[\frac{k}{2}+1,2a-t]$. These steps are detailed in Algorithm~\ref{alg1:v}.

\begin{algorithm}[!htbp]
  \caption{ Construction of $M_\mathcal{S}$ with $t\in[\lceil\frac{1+\sqrt{8k}}{2}\rceil,\frac{k}{2}+2]$ and $k$ is even.}
  \label{alg1:v}
  \begin{algorithmic}[1]
    \Require
      Integers $k,t$ with $t\in[\lceil\frac{1+\sqrt{8k}}{2}\rceil,\frac{k}{2}+2]$ and $k$ is even;
    \Ensure
      A sparse, good and balanced binary matrix $M_{\mathcal{S}}$.
    \State Compute $n=2k-t$, $a=\lceil\frac{k(k-1)}{2k-t}\rceil,b=\lfloor\frac{k(k-1)}{2k-t}\rfloor$ and $c=k+t-2a-1$.
    \label{alg1:v0}
    \State Construct the initial matrix $M_{\mathcal{S}}=(m_{i,j})$ as follows: for $i\in[1,\frac{k}{2}]$ and $j\in[i,i+k-2]$, $m_{i,j}=0$; for $i\in[\frac{k}{2}+1,k],j\in[1,a-\frac{k}{2}-1+i-\frac{k}{2}]\cup[\frac{k}{2}+a+1-t+i-\frac{k}{2},n]$, $m_{i,j}=0$; for all the rest positions, $m_{i,j}=1$.
    \label{alg1:v1}
    \State Change the positions of zeros.
    \label{alg1:v2}
    \State For $i\in[k-a+2,\frac{k}{2}],j\in[1,i-(k-a+1)]$ and $j'\in[2k-a,2k-a-1+i-(k-a+1)]$, $m_{i,j}=0$ and $m_{i,j'}=1$.
    \label{alg1:v3}
    \State For any two columns $j_1\in[k,2k-a-1]$ and $j_2 \in[2k-a,n]$, find a row $i\in[2,k-a]$ and $j_1$ as large as possible satisfying $m_{i,j_1}=0$ and $m_{i,j_2}=1$, then swap them:  $m_{i,j_1}=1$ and $m_{i,j_2}=0$.
    \label{alg1:v4}
    \State Repeat Step \ref{alg1:v4} until all columns from $[2k-a,n]$ have $a$ zeros and all columns from $[k+a-t+1,2k-a-1]$ have at least $b$ zeros, then we record the last $j_1$.
    \label{alg1:v5}
    \State If $j_1\geq k+a-t+1$, find a row $i\in[2a-t-\frac{k}{2}+2,\frac{k}{2}]$ and two columns $s_1\in[k+a-t+1,2k-a-1]$ and $s_2\in[2a-\frac{k}{2}-t,\frac{k}{2}-2]$ with  $s_2\leq i-2$, such that $m_{i,s_1}=0$ and $m_{i,s_2}=1$, then swap them:  $m_{i,s_1}=1$ and $m_{i,s_2}=0$.
    \label{alg1:v6}
    \State Repeat Step~\ref{alg1:v6} until all columns from $[k+a-t+1,2k-a-1]$ have at least $b$ zeros.
    \label{alg1:v7}
    \State If Steps~\ref{alg1:v6} and \ref{alg1:v7} are executed, then for $i\in[2a-t+2,k]$, find two columns $s_1\in[2a-\frac{k}{2}-t,\frac{k}{2}-2]$ and $s_2\in[\frac{k}{2}+1,2a-t]$, such that $m_{i,s_1}=0$ and $m_{i,s_2}=1$, then swap them:  $m_{i,s_1}=1$ and $m_{i,s_2}=0$.
    \label{alg1:v8}
    \State Repeat Step~\ref{alg1:v8} until all columns from $[2a-\frac{k}{2}-t,\frac{k}{2}-2]$ have $a$ zeros, and all columns from $[\frac{k}{2}+1,2a-t]$ have at most $a$ zeros.
    \label{alg1:v9}
    \State For $s_1\in[\frac{k}{2}+1,2a-t+1]$ and $s_2\in[2a-t+1,k+a-t]$, find a row $i\in[\frac{k}{2}+1,\frac{k}{2}+c]$, such that $m_{i,s_1}=1$ and $m_{i,s_2}=0$ and $(i,s_1)\neq (i,a-k+i)$, then swap them: $m_{i,s_1}=0$ and $m_{i,s_2}=1$.
    \label{alg1:v10}
    \State Repeat Step~\ref{alg1:v10} until all columns have $a$ or $b$ zeros.
    \label{alg1:v11}\\
    \Return $M_{\mathcal{S}}$;
  \end{algorithmic}
\end{algorithm}

\begin{remark}\label{rmkalg1}
  In Algorithm~\ref{alg1:v}, whether a column has $a$ or $b$ zeros depends on the integers $n,k,t$.  We can make a rule in advance that in the first few steps we try to obtain the required number (or close to) of columns containing $a$ zeros, then the remaining steps of the algorithm focus on making most of the rest columns to  have $b$ zeros. All algorithms in this section will follow this rule.
\end{remark}

\begin{lemma}\label{lemalg1}
 Algorithm~\ref{alg1:v} is executable and will terminate after finitely many iterations.
\end{lemma}

\begin{proof}
Algorithm~\ref{alg1:v} runs according to the agreement that all columns in $[1,\frac{k}{2}]\cup [2k-a,n]$ will have $a$ zeros, which is feasible since $\mu\geq k-\frac{t}{2}\geq \frac{k}{2}+a-t+1$. Our goal is to have $a$ or $b$ zeros in each column.
  Before Step~\ref{alg1:v4}, the number of extra zeros of all columns in $[k,k+a-t]$ or $[k+a-t+1,2k-a-1]$ is at least $\epsilon_1 =(k+t-2a-2)(a-t+1)$ or $\epsilon_2=\frac{(k+t-2a-2)(k+t-2a-1)}{2}$ respectively; the columns in $[2k-a,n]$ lacks $\iota$ zeros with $\iota=(a-t+1)(a-\frac{k}{2})$ in total. If the number of extras zeros of all columns in $[k+a-t+1,2k-a-1]$ is smaller than $\epsilon_2$, then the last  $j_1$ in Step~\ref{alg1:v5} satisfies $j_1<k+a-t+1$ and the program will skip Steps~\ref{alg1:v6}-\ref{alg1:v9} and go to Step~\ref{alg1:v10} directly, since $\epsilon_1+\epsilon_2\geq \iota$, which is obtained by  taking derivative with respect to $t$. Steps~\ref{alg1:v4} and \ref{alg1:v5} must terminate after $\iota$ iterations.



  If $\epsilon_2>\iota$,  Steps~\ref{alg1:v4} and \ref{alg1:v5} again terminate after $\iota$ iterations. Then the program will go through all Steps~\ref{alg1:v6}-\ref{alg1:v9} before Step~\ref{alg1:v10}. The difference between the number of extra zeros in columns $[k+a-t+1,2k-a-1]$ and the number of zeros that we need to move in columns $[2k-a,n]$ is at most $\delta$, where $\delta=\epsilon_2+k-2a+t-1-\iota=\frac{2a^2-3ak-2at+k^2+kt+t^2-t}{2}$.  Furthermore, before Step~\ref{alg1:v6}, $m_{i_1,s_2}=1$ for all $i_1\in[2a-t-\frac{k}{2}+2,\frac{k}{2}]$ and $s_2\in[2a-\frac{k}{2}-t,\frac{k}{2}-2]$ with $s_2\leq i_1-2$; before Step~\ref{alg1:v8}, all columns have at least $a$ zeros except columns in $[\frac{k}{2}+1,2a-t]$. Thus, Steps~\ref{alg1:v6}, \ref{alg1:v7}, \ref{alg1:v8} and \ref{alg1:v9} must terminate after $\delta$ iterations.


  Before Step~\ref{alg1:v10}, the number of zeros in the $j$th column is at most $a$ if $j\in[\frac{k}{2}+1,2a-t+1]$, and at least $a$ if $j\in[2a-t+1,k+a-t]$. We only need to adjust the positions of zeros in these columns to make the matrix balanced. Furthermore, for $i\in[\frac{k}{2}+1,\frac{k}{2}+c],j\in[2a-t+1,k+a-t]$, there are enough zeros to move, we can find a row $i$ to complete Step~\ref{alg1:v10}. Hence, Steps~\ref{alg1:v10} and \ref{alg1:v11} must terminate after finitely many iterations.
\end{proof}


\begin{lemma}
  Algorithm~\ref{alg1:v} returns a good $M_{\mathcal{S}}$.
\end{lemma}

\begin{proof}
  The initial matrix $M_{\mathcal{S}}$ corresponds to a good binary tree $(\frac{k}{2};\frac{k}{2}-1,k-1;\frac{k}{2}-2,k-2;\ldots;1,\frac{k}{2}+1)$. In  Steps~\ref{alg1:v3}, \ref{alg1:v4} and \ref{alg1:v6}, we only refine the $(i,j)$th position with $i\in[2,\frac{k}{2}],j\in [1,a-\frac{k}{2}-1]\cup [k,n]$. By Lemma~\ref{thm_base},  the new $\mathcal{S}$ is still good with the same tree. Similarly, in Steps~\ref{alg1:v8} and \ref{alg1:v10}, for the $i$th row with $i\in[2a-t+2,k]$, we only refine the columns in $[2a-\frac{k}{2}-t,2a-t]$; for the $i$th row with $i\in[\frac{k}{2}+1,\frac{k}{2}+c]$, we only refine the columns in $[\frac{k}{2}+1,k+a-t]$. In fact,  the maximum number of extra zeros in columns $[2a-t+1,k+a-t-1]$ is at most $c$, where $c$ is defined in  Algorithm~\ref{alg1:v}. Since $\frac{k}{2}+c<2a-t+2$,  we see that each step in Algorithm~\ref{alg1:v} does some operations allowed in Lemma~\ref{thm_base}, and the final new $\mathcal{S}$ still corresponds to a good  binary tree $(\frac{k}{2};\frac{k}{2}-1,\frac{k}{2}+c;\frac{k}{2}-2,\frac{k}{2}+1,k-1;\frac{k}{2}-3,\frac{k}{2}+2,k-2;\ldots)$.
\end{proof}

\subsubsection{When $t\in[\frac{k}{2}+3, \lfloor k-\sqrt{k}\rfloor]$}
\
\newline
\indent In this case, $n<\frac{3k}{2}-2$ and $a-1\geq \frac{k}{2}+a+2-t$. Start from an initial matrix $M_{\S}=[M_1,M_2]^\top$, where $M_1$ and $M_2$ are defined in the last subsection.  See Fig.~\ref{algidea2} about the zero positions of $M_{\mathcal{S}}$, which is slightly different from Fig.~\ref{algidea} due to a small $n$.

\begin{figure}[!htbp]
\centering
\includegraphics[scale=0.6]{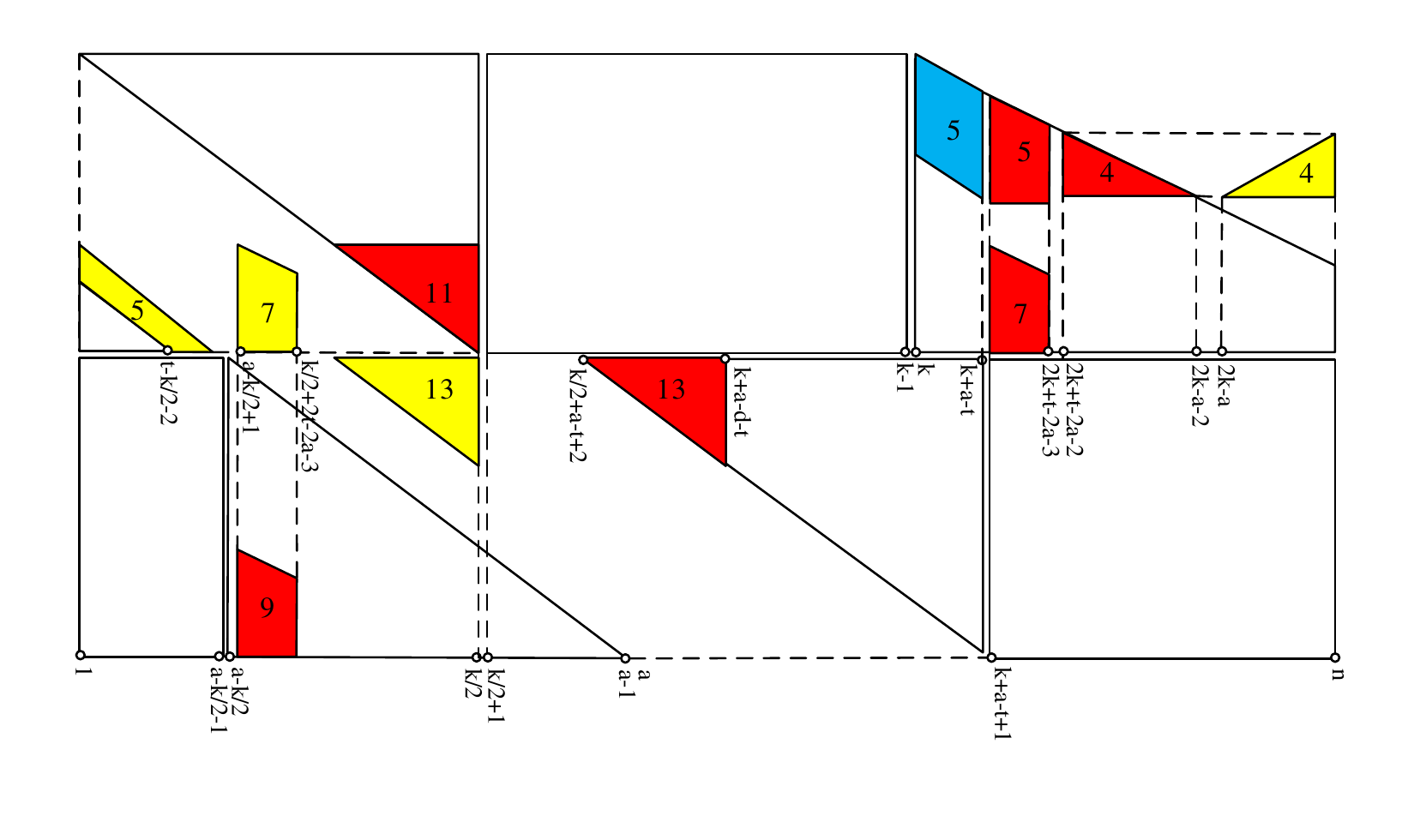}
\caption{The illustration of Algorithm~\ref{alg2:u}. The indices of blocks corresponds to the  indices of Steps in Algorithm~\ref{alg2:u}.}\label{algidea2}
\end{figure}

 Any column of $M_{\S}$ in $[a-\frac{k}{2},\frac{k}{2}]$ has $a$ zeros, which is the same as before, but each column with index in $[1,t-\frac{k}{2}-2]$ has $t-1$ zeros which is different from before. Further, $M_{\S}$ has a smaller number of columns, but it contains the same amount of zeros as before, so each column in $M_{\S}$ will have more zeros. Since $\mu\geq\frac{n}{2}\geq k-\frac{t}{2}\geq k+a-2t+4$, we can make all columns in $[a-\frac{k}{2},\frac{k}{2}]\cup [2k+t-2a-2,n]$ to have $a$ zeros. The $(2k-a-1)$th column already has exactly $a$ zeros in the initial matrix. For columns from $[2k-a,n]$,  we still need some zeros like the yellow triangle labeled by $4$ in Fig.~\ref{algidea2}, which can be moved from columns $[2k+t-2a-2,2k-a-2]$ like the red triangle labeled by $4$.  Unlike Algorithm~\ref{alg1:v}, we can not find a regular shape of zeros in the last $k-a$ columns of $M_{\S}$ to fill in the first $a-\frac{k}{2}-1$ columns, but we can move the zeros like the red echelon labeled by $5$ to the yellow $5$.

If the number of zeros of the red $5$ is more than the yellow $5$, after the above modification, all columns in $[1,\frac{k}{2}]\cup[2k+t-2a-2,n]$ have $a$ or $b$ zeros. There are still some columns in $[k+a-t+1,2k+t-2a-3]$ having extra zeros. These extra zeros can not be moved to any other column in row $[1,\frac{k}{2}]$ directly since it will increase extra zeros to  columns in $[1,\frac{k}{2}]\cup[2k+t-2a-2,n]$; and they also can not be moved to any other column in rows $[\frac{k}{2}+1,n]$, otherwise it may destroy the good condition of the bottom block since they are not operations allowed in Lemma~\ref{thm_base}.  To make the matrix balanced while keeping the good property, we first move the extra zeros in columns $[k+a-t+1,2k+t-2a-3]$ to columns $[a-\frac{k}{2}+1,\frac{k}{2}+2t-2a-3]$. This modification is illustrated in Fig.~\ref{algidea2} by moving some zeros in the red echelon labeled by $7$ to the yellow $7$. Then the columns in $[a-\frac{k}{2}+1,\frac{k}{2}+2t-2a-3]$ have extra zeros. We move these extra zeros identified by the red echelon labeled by $9$ in Fig.~\ref{algidea2} to columns in $[\frac{k}{2}+1,2a-t]$.

The challenge of our approach is when the number of zeros of the red $5$ is less than the yellow $5$. We need to bring extra zeros from the $j$th column with $j\in[k,k+a-t]$, like the blue echelon labeled by $5$. The worst thing is that the number of zeros of the red $5$ and blue $5$ is still less than the yellow $5$. Then we bring extra zeros from the red triangle labeled by $11$.  To compensate the red triangle $11$ for zeros, we  move in the yellow triangle labeled by $13$ from the red triangle  $13$. These steps  are detailed in Algorithm~\ref{alg2:u}.

\begin{algorithm}[!htbp]
  \caption{ Construction of $M_\mathcal{S}$ with $\frac{k}{2}+3\leq t\leq \lfloor k-\sqrt{k}\rfloor$ and $k$ is even.}
  \label{alg2:u}
  \begin{algorithmic}[1]
    \Require
      Integers $k,t$ with $\frac{k}{2}+3\leq t\leq \lfloor k-\sqrt{k}\rfloor$ and $k$ is even;
    \Ensure
      A sparse, good and balanced binary matrix $M_{\mathcal{S}}$.
    \State Compute $n=2k-t$, $a=\lceil\frac{k(k-1)}{2k-t}\rceil,b=\lfloor\frac{k(k-1)}{2k-t}\rfloor, c=k-2a+t-1$ and $d=\max\{a-\frac{k}{2}+2,k-a+1\}$.
    \label{alg2:u0}
    \State Construct the initial matrix $M_{\mathcal{S}}=(m_{i,j})$ as follows: for $i\in[1,\frac{k}{2}]$ and $j\in[i,i+k-2]\mm n$, $m_{i,j}=0$; for $i\in[\frac{k}{2}+1,k],j\in[1,a-\frac{k}{2}-1+i-\frac{k}{2}]\cup[\frac{k}{2}+a+1-t+i-\frac{k}{2},n]$, $m_{i,j}=0$; for all the rest positions, $m_{i,j}=1$.
    \label{alg2:u1}
    \State Change the positions of zeros.
    \label{alg2:u2}
    \State For $i\in[k+t-2a,k-a],s_1\in[2k+t-2a-2,2k+t-2a-2+i-(k+t-2a)]$ and $s_2\in[n-(i-(k+t-2a)),n]$, $m_{i,s_1}=1$ and $m_{i,s_2}=0$.
    \label{alg2:u3}
    \State For any two columns $s_1\in[1,a-\frac{k}{2}-1]$ and $s_2 \in[k,2k+t-2a-3]$, find a row $i\in[3,k-a]$ with $s_2$ as large as possible but $s_1$ as small as possible, and then $i$ is as small as possible satisfying $m_{i,s_1}=1$ and $m_{i,s_2}=0$ and $(i,s_1)\neq (s_1+1,s_1)$, then swap them:  $m_{i,s_1}=0$ and $m_{i,s_2}=1$.
    \label{alg2:u4}
    \State Repeat Step~\ref{alg2:u4} until all columns from $[1,a-\frac{k}{2}-1]$ have $a$ or $b$ zeros, then we record the last $s_2$ and go to Step~\ref{alg2:u6}; Or until all columns from $[k,2k+t-2a-3]$ have $a$ or $b$ zeros, then we go to Step~\ref{alg2:u10}.
    \label{alg2:u5}
    \State If the $s_2$ in Step~\ref{alg2:u5} satisfies $s_2\geq k+a-t+1$, find a row $i\in[2a-t-\frac{k}{2}+2,\frac{k}{2}]$ and two columns $s_1\in[a-\frac{k}{2}+1,\frac{k}{2}+2t-2a-3]$ and $s_2\in[k+a-t+1,2k+t-2a-3]$, such that $m_{i,s_1}=1$ and $m_{i,s_2}=0$ and $(i,s_1)\neq (s_1+1,s_1)$, then swap them:  $m_{i,s_1}=0$ and $m_{i,s_2}=1$.
    \label{alg2:u6}
    \State Repeat Step~\ref{alg2:u6} until all columns from $[k+a-t+1,n]$ have $a$ or $b$ zeros.
    \label{alg2:u7}
    \State If Steps~\ref{alg1:v6} and \ref{alg1:v7} are executed, then for $i\in[2a-t+2,k]$, find two columns $s_1\in[a-\frac{k}{2}+1,\frac{k}{2}+2t-2a-3]$ and $s_2\in[\frac{k}{2}+1,2a-t]$, such that $m_{i,s_1}=0$ and $m_{i,s_2}=1$, then swapping them: set $m_{i,s_1}=1$ and $m_{i,s_2}=0$.
    \label{alg2:u8}
    \State Repeat Step~\ref{alg2:u8} until all columns from $[a-\frac{k}{2}+1,\frac{k}{2}+2t-2a-3]$ have $a$ zeros, and all columns from $[\frac{k}{2}+1,2a-t]$ have at most $a$  zeros, then turn to Step~\ref{alg2:u14}.
    \label{alg2:u9}
    \State For any two columns $s_1\in[1,a-\frac{k}{2}-1]$ and $s_2\in[d,\frac{k}{2}-1]$, find a row $i\in[d,\frac{k}{2}-1]$ as large as possible satisfying $m_{i,s_1}=1$ and $m_{i,s_2}=0$ with $(i,s_1)\neq(s_1+1,s_1)$, then swap them: $m_{i,s_1}=0$ and $m_{i,s_2}=1$.
    \label{alg2:u10}
    \State Repeat Step~\ref{alg2:u10} until all columns from $[1,a-\frac{k}{2}-1]$ have $a$ or $b$ zeros.
    \label{alg2:u11}
    \State If Steps~\ref{alg2:u10} and \ref{alg2:u11} are executed, then for $s_1\in[d,\frac{k}{2}-1]$ and $s_2\in[\frac{k}{2}+a+2-t,k+a-d-t+1]$, find a row $i\in[\frac{k}{2}+1,k-d]$ as small as possible and $s_1\geq a-k+i+1$, such that $m_{i,s_1}=1$ and $m_{i,s_2}=0$, then swap them:  $m_{i,s_1}=0$ and $m_{i,s_2}=1$.
    \label{alg2:u12}
    \State Repeat Step~\ref{alg2:u12} until all columns in $[d,\frac{k}{2}-1]$ have $a$  zeros.
    \label{alg2:u13}
    \State For $s_1\in[\frac{k}{2}+1,2a-t+1]$ and $s_2\in[2a-t+1,k+a-t]$, find a row $i\in[\frac{k}{2}+1,\frac{k}{2}+c]$, such that $m_{i,s_1}=1$ and $m_{i,s_2}=0$, then swap them:  $m_{i,s_1}=0$ and $m_{i,s_2}=1$.
    \label{alg2:u14}
    \State Repeat Step~\ref{alg2:u14} until all columns have $a$ or $b$ zeros.
    \label{alg2:u15}\\
    \Return $M_{\mathcal{S}}$;
  \end{algorithmic}
\end{algorithm}


\begin{lemma}\label{lemalg2}  Algorithm~\ref{alg2:u} is executable and will terminate after finitely many iterations.
\end{lemma}

\begin{proof}
 Algorithm~\ref{alg2:u} runs according to the rule that  all columns in $[a-\frac{k}{2},\frac{k}{2}]\cup [2k+t-2a-2,n]$ will have $a$ zeros, which is feasible since $\mu\geq\frac{n}{2}\geq k-\frac{t}{2}\geq k+a-2t+4$.

  After Step~\ref{alg2:u3}, any column in $[2k+t-2a-2,n]$ does not have extra zeros.
  Before Step~\ref{alg2:u4}, the number of extra zeros of all columns in $[k,k+a-t]$ or $[k+a-t+1,2k+t-2a-3]$ is at least $\epsilon_1=(k+t-2a-2)(a-t+1)$ or $\epsilon_2=\frac{3a-3k-4ak-2at+2kt+3a^2+k^2}{2}$ respectively; the columns in $[1,a-\frac{k}{2}-1]$ lacks at most $\iota=\frac{3t-k-a-ka+kt+a^2-t^2-2}{2}$ zeros in total. If the number of extra zeros of columns in $[k,2k+t-2a-3]$ is no more than the zeros that  columns $[1,a-\frac{k}{2}-1]$ need, Steps~\ref{alg2:u4} and \ref{alg2:u5} will terminate after at most $\epsilon_1+a-t+1+\epsilon_2+k+2t-3a-3=\epsilon_1+\epsilon_2+k+t-2a-2$
    iterations.

    After that, all columns from $[k,2k+t-2a-3]$ have $a$ or $b$ zeros, and then go to Step~\ref{alg2:u10} directly. There are  $\frac{(k/2-d)(k/2-d+1)}{2}$ zeros in the columns from $d$ to  $(\frac{k}{2}-1)$ and in the rows from $d$ to $(\frac{k}{2}-1)$, where $d$ is defined in Algorithm~\ref{alg2:u}. If $d=k-a+1$, then $\frac{k}{2}-d=a-\frac{k}{2}-1$. Thus we have enough zeros to be moved out.
  If $d=a-\frac{k}{2}+2$,  there are $\frac{(k-a-2)(k-a-1)}{2}$ zeros in the columns from $(a-\frac{k}{2}+2)$ to $(\frac{k}{2}-1)$ and in the rows from $(a-\frac{k}{2}+2)$ to $(\frac{k}{2}-1)$. Then the number of extra zeros in the red $5$, blue $5$ and red $11$ is bigger than that in the yellow $5$, since
  \begin{align*}
    &\epsilon_1+\epsilon_2+ \frac{(k-a-2)(k-a-1)}{2}-\iota\\
    =&\frac{(k+a-t)(k-a)-(3a-k-t)(k-t)-a-3k+3t}{2}>0.
\end{align*}
The above inequality is obtained by assuming $a=\frac{k}{2}+\frac{t-2}{4}$ in the term $k+a-t$, and $a=\frac{k}{2}+\frac{t-1}{2}$ in the rest.
  By the above analysis, Steps~\ref{alg2:u10} and \ref{alg2:u11} must terminate after at most $\frac{(k/2-d-1)(k/2-d)}{2}$ iterations, so do Steps~\ref{alg2:u12} and \ref{alg2:u13}.


  If $\epsilon_2> \iota$, Steps~\ref{alg2:u4} and \ref{alg2:u5} must terminate after $\iota$ iterations, then the program will go to Steps~\ref{alg2:u6}, \ref{alg2:u7}, \ref{alg2:u8} and \ref{alg2:u9}. These steps will terminate after at most

  \[\epsilon_2-\iota+a-\frac{k}{2}-1=\frac{2a^2 - 3ak - 2at + 6a + k^2 + kt - 3k + t^2 - 3t }{2}\]
iterations, which equals the maximum  possible number of extra zeros in  columns $[k+a-t+1,2k+t-2a-3]$ before Step~\ref{alg2:u6}.

  After Step~\ref{alg2:u13}, the number of zeros in the $j$th column is at most $a$ if $j\in[\frac{k}{2}+1,2a-t+1]$, at least $a$ if $j\in[2a-t+1,k+a-t]$, and exactly $a$ or $b$  for all the rest columns. The goal of Steps~\ref{alg2:u14} and \ref{alg2:u15} is to make all columns in $[\frac{k}{2}+1,k+a-t]$ balanced, which will terminate after finitely many iterations.
\end{proof}

\begin{lemma}
  Algorithm~\ref{alg2:u} returns a good $M_{\mathcal{S}}$.
\end{lemma}

\begin{proof}
  The initial matrix $M_{\mathcal{S}}$ corresponds to a good binary tree $(\frac{k}{2};\frac{k}{2}-1,k-1;\frac{k}{2}-2,k-2;\ldots;1,\frac{k}{2}+1)$. In Steps~\ref{alg1:v3} and \ref{alg2:u4}, from the $2$th row to the $(k-a)$th row, we only refine columns in $[k,n]$. In Step~\ref{alg2:u6}, we only refine the $(i,j)$th position with $i\in[2a-t-\frac{k}{2}+2,\frac{k}{2}]$ and $j\in[a-\frac{k}{2}+1,\frac{k}{2}+2t-2a-3]\cup[k+a-t+1,2k+t-2a-3]$. By Lemma~\ref{thm_base}, the new $\mathcal{S}$ is still good with the same tree. In Step~\ref{alg2:u10}, from the $d$th row to the $(\frac{k}{2}-1)$th row, we only refine columns in $[d,\frac{k}{2}-1]$, where $d$ is defined in Algorithm~\ref{alg2:u}. Since $d\geq k-a+1$,  the new $\mathcal{S}$  still corresponds to a good binary tree $(\frac{k}{2};d,k-1;d-1,d+1,k-2;d-2,d+2,k-3;\ldots)$.

  Similarly, in Step~\ref{alg2:u8}, for  rows in $[2a-t+2,k]$, we only refine  columns in $[a-\frac{k}{2}+1,\frac{k}{2}+2t-2a-3]\cup[\frac{k}{2}+1,2a-t]$; in Step~\ref{alg2:u12}, for rows in $[\frac{k}{2}+1,k-d]$, we only refine  columns in $[d,\frac{k}{2}-1]\cup[\frac{k}{2}+a+2-t,k+a-d-t]$; in Step~\ref{alg2:u14}, for  rows in $[\frac{k}{2}+1,\frac{k}{2}+c]$, we only refine  columns in $[\frac{k}{2}+1,k+a-t]$. Notice that only one of Steps~\ref{alg2:u8} and \ref{alg2:u12} will be performed. If Steps~\ref{alg2:u12} and \ref{alg2:u14} are executed, the final new $\mathcal{S}$  corresponds to a good binary tree $(\frac{k}{2};d,\frac{k}{2}+c;d-1,d+1,\frac{k}{2}+1,k-1;d-2,d+2,\frac{k}{2}+2,k-2;\ldots)$. If Steps~\ref{alg2:u8} and \ref{alg2:u14} are executed, since $\frac{k}{2}+c<2a-t+2$, it is easy to see each step in Algorithm~\ref{alg2:u} does some operations allowed in Lemma~\ref{thm_base}, and the finial new $\mathcal{S}$  corresponds to a good binary tree $(\frac{k}{2};\frac{k}{2}-1,\frac{k}{2}+c;\frac{k}{2}-2,\frac{k}{2}+1,k-1;\frac{k}{2}-3,\frac{k}{2}+2,k-2;\ldots)$.
\end{proof}

\subsubsection{When $t\in[\lceil k-\sqrt{k}\rceil, k]$}
\
\newline
\indent For $\lceil k-\sqrt{k}\rceil\leq t\leq k-1$, we give a little more explicit construction, see Algorithm~\ref{alg3:w}. When $t=k$, the matrix $M_{\mathcal{S}}$ is given by setting all entries $1$ except the diagonal entries. Let $u=k-t$, then $1\leq u\leq \lfloor \sqrt{k}\rfloor$ and $n=k+u$. Since


\[\frac{k(k-1)}{k+u}=k-\frac{k(u+1)}{k+u}=k-u-1+\frac{u(u+1)}{k+u},\]
we have $a=k-u$, $b=k-u-1$ and $\mu= u(u+1)$. 

\begin{algorithm}[!htbp]
  \caption{ Construction of $M_\mathcal{S}$ with  $\lceil k-\sqrt{k}\rceil\leq t\leq k$ and $k$ is even.}
  \label{alg3:w}
  \begin{algorithmic}[1]
    \Require
      Integers $k,u$ with $1\leq u\leq \lfloor \sqrt{k}\rfloor$ and $k$ is even;
    \Ensure
      A sparse, good and balanced binary matrix $M_{\mathcal{S}}$.
    \State Compute $n=k+u$, $a=k-u, b=k-u-1$.
    \label{alg2:w0}
    \State Construct the initial matrix $M_{\mathcal{S}}=(m_{i,j})$ as follows: for $i\in[1,k]$ and $j\in[i,i+k-2]\mm n$, $m_{i,j}=0$; for all the rest positions, $m_{i,j}=1$.
    \label{alg3:w1}
    \State Change the positions of zeros.
    \label{alg3:w2}
    \State For any two columns in $s_1\in[1,\frac{k}{2}-2]$ and $s_2\in[k,n-1]$, find a row $i\in[3,\frac{k}{2}]$ as large as possible and $i\geq s_1+2$ satisfying $m_{i,s_1}=1$ and $m_{i,s_2}=0$, then swap them: $m_{i,s_1}=0$ and $m_{i,s_2}=1$.
    \label{alg3:w3}
    \State Repeat Step~\ref{alg3:w3} until all columns from $[1,a-\frac{k}{2}-1]$ have $a$ or $b$ zeros, and all columns from $[k,n-1]$ have $a$ zeros.
    \label{alg3:w4}
    \State For any two columns $s_1\in[\frac{k}{2}-1,a-1]$ and $s_2 \in[a,k-1]$, find a row $i\in[\frac{k}{2}+1,k]$  satisfying $m_{i,s_1}=1$ and $m_{i,s_2}=0$ and $s_1\geq i-u$, then swap them: $m_{i,s_1}=0$ and $m_{i,s_2}=1$.
    \label{alg3:w5}
    \State Repeat Step~\ref{alg3:w5} until all columns  have $a$ or $b$ zeros.
    \label{alg3:w6}\\
    \Return $M_{\mathcal{S}}$;
  \end{algorithmic}
\end{algorithm}

\begin{lemma}\label{lemalg3}
  Algorithm~\ref{alg3:w} returns a good $M_{\mathcal{S}}$.
\end{lemma}

\begin{proof} Algorithm~\ref{alg3:w} runs according to the rule that  all columns in $[a,k-1]\cup [k,n-1]$ will have $a$ zeros, which is feasible since $\mu=u(u+1)\geq 2u$. The initial matrix  $M_{\mathcal{S}}$ in Algorithm~\ref{alg3:w} is a circulant matrix with the first row corresponding to the set $[1,k-1]$.
The $j$th column of $M_{\mathcal{S}}$ has $b$ zeros if $j\in [1,a-1]$, $j$ zeros if $j\in[a,k-1]$ and  $2k-j-1$ zeros if $j\in[k,n]$.
Observe that  columns in $[a,k-1]$ have  $\frac{u(u-1)}{2}$ extra zeros in total,   and so do columns in $[k,k+u-1]$.

  In Step~\ref{alg3:w3}, we  remove the extra zeros in columns $[k,n-1]$ to columns $[1,\frac{k}{2}-2]$. Notice that we can only move one zero to each column in $[1,\frac{k}{2}-2]$. Since $\frac{u(u-1)}{2}\leq \frac{k-\sqrt{k}}{2}\leq\frac{k}{2}-2$ for $k\geq 16$,  Step~\ref{alg3:w4} will terminate after $\frac{u(u-1)}{2}$ iterations.  By Lemma~\ref{thm_base}, the new $M_{\mathcal{S}}$ still corresponds to a good binary tree $(\frac{k}{2};\frac{k}{2}-1,\frac{k}{2}+1;\frac{k}{2}-2,\frac{k}{2}+2;\ldots;1,k-1)$. Same analysis to Steps~\ref{alg3:w5} and \ref{alg3:w6}, and the good binary tree does not change. When $k<16$, it is easy to construct a sparse, good and balanced $M_{\mathcal{S}}$.
\end{proof}

\subsection{The Complexity of the Algorithms}

In Algorithms~\ref{alg:u2=1}--\ref{alg3:w}, each initial matrix $M_{\mathcal{S}}$ satisfies the good condition, and each step satisfies the conditions of Lemmas~\ref{thm_base}, so we do not need to take time for the verification of the good condition. According to the proofs of Lemmas~\ref{lemalg1}, \ref{lemalg2}, and \ref{lemalg3}, and the analysis of Algorithm~\ref{alg:u2=1}, it is obvious that Algorithms~\ref{alg:u2=1}--\ref{alg3:w} run in polynomial time in $k$ and $n$.

\section{Conclusion}\label{con}

To conclude, we first  present a new sufficient constraint on the zero patterns so that an $[n,k]_q$ MDS code exists with a sparse generator matrix satisfying the given zero pattern for all $q\geq n$. Then based on this constraint, we construct an $[n,k]_q$ MDS code with a sparse and balanced generator matrix for all $q\geq n$ provided that $n\leq 2k$, by designing several polynomial-time algorithms in $k$ and $n$. The condition $n\leq 2k$ is restricted from the sufficient constraint of the  zero patterns. So we need a new method to deal with the case when $n> 2k$.
Further, it is interesting to consider a smaller field size, that is $q=n-1$. We leave these problems for future study.

\section*{Appendix}

\begin{lemma*}[Lemma~\ref{SBGM1}]
  Suppose that the $(k-1)$-uniform set system $\mathcal{S}=\{S_1,S_2,\ldots,S_k\}$  is non-intersecting and separable at some $i\in[k-1]$. Then
  \begin{gather*}
    C(\P)=\begin{bmatrix}
      C(\P_1) & 0\\
      0 & C(\P_2)
    \end{bmatrix}
    \begin{bmatrix}
      C(\Q_1)\\
      C(\Q_2)
    \end{bmatrix}.
  \end{gather*}
\end{lemma*}

\begin{proof}
  For $j\in[i]$, let $P_{S_j'}(x)=\prod_{\ell\in S_j^\prime}(x-a_\ell)=c_{j,0}x^{i-1}+c_{j,1}x^{i-2}+\cdots+c_{j,i-1}$. Then
  \begin{align*}
    P_{S_j}=&f_0\cdot P_{S_j'}=c_{j,0}x^{i-1}f_0+c_{j,1}x^{i-2}f_0+\cdots+c_{j,i-1}f_0.
  \end{align*}
  So the coefficient of $x^{e}$ in $P_{S_j}$ is $[x^e]P_{S_j}=\sum_{\ell=0}^{i-1} c_{j,\ell} \times [x^{e}](x^{i-\ell-1}f_0)$, $e\in [0,k-1]$. Since $[x^e]P_{S_j}$ is the $(j,k-e)$th entry of $C(\{P_{S_1},\ldots,P_{S_i}\})$, then

  \begin{align*}
    &C(\{P_{S_1},\ldots,P_{S_i}\})\\
    =&\begin{bmatrix}
      c_{1,0} &\cdots &c_{1,i-1}\\
      c_{2,0} &\cdots &c_{2,i-1}\\
      \vdots&& \vdots\\
      c_{i,0} &\cdots &c_{i,i-1}
    \end{bmatrix}
    \begin{bmatrix}
      [x^{k-1}](x^{i-1}f_0) & \cdots & [x^0](x^{i-1}f_0)\\
      [x^{k-1}](x^{i-2}f_0) & \cdots & [x^0](x^{i-2}f_0)\\
      \vdots&& \vdots\\
      [x^{k-1}]f_0 & \cdots & [x^0]f_0\\
    \end{bmatrix}\\
    =&C(\P_1)C(\Q_1).
  \end{align*}
  The bottom part is similar. We can get $C(\{P_{S_{i+1}},\ldots,P_{S_k}\})=C(\P_2)C(\Q_2)$. This completes the proof by noting that $C(\P)=(C(\{P_{S_1},\ldots,P_{S_i}\}) \; C(\{P_{S_{i+1}},\ldots,P_{S_k}\}))^T$.
\end{proof}

%
%

\begin{lemma*}[Lemma~\ref{SBGM2}]
  The determinant of $[C(\Q_1)\; C(\Q_2)]^T$ is nonzero. In particular  $\det([C(\Q_1)$ $C(\Q_2)]^T)=\prod_{u\in A,v\in B}(a_u-a_v)$.
\end{lemma*}

\begin{proof} Let $A=\{u_1,\ldots,u_{k-i}\}$ and $B=\{v_1,\ldots,v_i\}$.
  Let $f_s=(x-a_{u_{s+1}})\cdots(x-a_{u_{k-i}})=c_{s,0}x^{k-i-s}+\cdots+c_{s,k-i-s}, s\in[k-i]$ and $g_t=(x-a_{v_{t+1}})\cdots(x-a_{v_{i}})=d_{t,0}x^{i-t}+\cdots+d_{t,i-t}, t\in[i]$. Note that $f_s$ and $g_t$ can be obtained from $f_0$ and $g_0$ by deleting some linear factors. Further, $c_{s,0}=d_{t,0}=1$ for all $s\in[k-i]$ and $t\in[i]$,
 $f_{k-i}=g_i=1$.

Let  $M_0=[C(\Q_1)$ $C(\Q_2)]^T=[C(x^{i-1}f_0)$ $\cdots$ $C(f_0)$ $C(x^{k-i-1}g_0)$ $\cdots$ $C(g_0)]^T$. We compute the determinant of $M_0$ by doing elementary row operations. Since each row corresponds to a polynomial, we use polynomial operations to consider row operations.  For convenience, let $R^\ell$, $\ell\in [i]$ denote the  row $C(x^{\ell-1}f_0)$, and $R_\ell$, $\ell\in [k-i]$ denote the row $C(x^{\ell-1}g_0)$.

  Step $1$.  Change $R_{k-i}$ to $c_{1,0}R_{k-i}+\cdots+c_{1,k-i-1}R_1-(d_{1,0}R^i+\cdots+d_{1,i-1}R^1)$. Remember that $c_{1,0}=1$. Then the polynomial corresponding to this row becomes
  \begin{align*}
    & \sum_{j=0}^{k-i-1}c_{1,j}x^{k-i-1-j}g_0-\sum_{j=0}^{i-1}d_{1,j}x^{i-1-j}f_0\\
    =&f_1g_0-f_0g_1=(x-a_{v_1})f_1g_1-(x-a_{u_1})f_1g_1\\
    =&(a_{u_1}-a_{v_1})f_1g_1.\\
  \end{align*}
  Continuing row operations to the new $R_{k-i}$ by subtracting $(a_{u_1}-a_{v_1})(d_{2,0}R^{i-1}+\cdots+d_{2,i-2}R^1)$, the corresponding polynomial becomes
  \begin{align*}
    &f_1g_0-f_0g_1-(a_{u_1}-a_{v_1})f_0g_2\\
    =&(a_{u_1}-a_{v_1})(x-a_{v_2})f_1g_2-(a_{u_1}-a_{v_1})(x-a_{u_1})f_1g_2\\
    =&(a_{u_1}-a_{v_1})(a_{u_1}-a_{v_2})f_1g_2.
  \end{align*}
  Repeat similar row operations to $R_{k-i}$, we obtain a polynomial
  \[f_1g_0-f_0g_1-\sum_{t=1}^{i-1}(\prod_{s=1}^t(a_{u_1}-a_{v_s}))f_0g_{t+1}=f_1\prod_{t=1}^i(a_{u_1}-a_{v_t}).\]
  Hence we change the row $C(x^{k-i-1}g_0)$ to $C(f_1\prod_{t=1}^i(a_{u_1}-a_{v_t}))$ without changing the determinant. We extract the nonzero factors $\prod_{t=1}^i(a_{u_1}-a_{v_t})$ and assume the new row is $C(f_1)$.
  Observe that  for each $j=0,1,\cdots,i-1$, $x^jf_0+x^ja_{u_1}f_1=x^j(x-a_{u_1})f_1+x^ja_{u_1}f_1=x^{j+1}f_1$. Then we can do a sequence of row operations: add $a_{u_1}\times R_{k-i}$ to $R^1$, add  $a_{u_1}\times R^{\ell}$ to $R^{\ell+1}$,  $\ell=1,\ldots,i-1$. Then after these operations, the matrix $M_0$ is changed to $M_1=\prod_{t=1}^i(a_{u_1}-a_{v_t}))[C(x^{i}f_1)\; \cdots\; C(f_1)\; C(x^{k-i-2}g_0) \;\cdots\; C(g_0)]^T$.

Step $2$. Note that the matrix $M_1$ has a similar pattern with the original matrix $M_0$. So we can update our row notations as follows.  Let $R^\ell$, $\ell\in [i+1]$ denote the  row $C(x^{\ell-1}f_1)$, and $R_\ell$, $\ell\in [k-i-1]$ denote the row $C(x^{\ell-1}g_0)$. Under this new notation, we do similar row operations to $R_{k-i-1}$ as Step $1$, to get $M_2=\prod_{t\in[i],j\in[2]}(a_{u_j}-a_{v_t})[C(x^{i+1}f_2)$ $\cdots$ $C(f_2)$ $C(x^{k-i-3}g_0) $ $\cdots$ $ C(g_0)]^T$. We illustrate these operations in the $r$th step.

  Step $r\leq k-i$. Now we have a matrix $M_{r-1}=\prod_{t\in[i],j\in[r-1]}$
  $(a_{u_j}-a_{v_t})[C(x^{i+r-2}f_{r-1})$ $\cdots$ $C(f_{r-1})$ $ C(x^{k-i-r}g_0)$ $\cdots$ $C(g_0)]^T$.  Update the row notations as follows. Let $R^\ell$, $\ell\in [i+r-1]$ denote the  row $C(x^{\ell-1}f_{r-1})$, and $R_\ell$, $\ell\in [k-i-r+1]$ denote the row $C(x^{\ell-1}g_0)$. Consider the row $R_{k-i-r+1}$ which  corresponds to $C(x^{k-i-r}g_0)$,  do the row operations $c_{r,0}R_{k-i-r+1}+\cdots+c_{r,k-i-r}R_1-(d_{1,0}R^i+\cdots+d_{1,i-1}R^1)$ first to get $f_rg_0-f_{r-1}g_1$, then continue row operations to the new $R_{k-i-r+1}$ by subtracting $\sum_{t=1}^{i-1}(\prod_{s=1}^t(a_{u_r}-a_{v_s}))(d_{t+1,0}R^{i-t}+\cdots+d_{t+1,i-t-1}R^1)$ to get
  \begin{equation*}
    f_r g_0-f_{r-1}g_1-\sum_{t=1}^{i-1}(\prod_{s=1}^t(a_{u_r}-a_{v_s}))f_{r-1}g_{t+1}=f_r\prod_{t=1}^i(a_{u_r}-a_{v_t}).
  \end{equation*}
  Hence we update the row $C(x^{k-i-r}g_0)$ to $C(f_r\prod_{t=1}^i(a_{u_r}-a_{v_t}))$ without changing the determinant. Observe that for $0\leq j\leq i+r-2$, $x^{j+1}f_r=x^jf_{r-1}+a_{u_r}x^jf_r$. Then we can do a sequence of the following operations: add $a_{u_r}\times R_{k-i-r+1}$ to $R^1$ and add  $a_{u_r}\times R^{\ell}$ to $R^{\ell+1}$,  $\ell=1,\ldots,i+r-2$. After  the $r$th step, we change the matrix $M_{r-1}$ in Step $(r-1)$ to the following,

    \[\prod_{t\in[i],j\in[r]}(a_{u_j}-a_{v_t})[C(x^{i+r-1}f_r)\;\cdots\;C(f_r)\;C(x^{k-i-r-1}g_0)\;\cdots\;C(g_0)]^T.\]

  After $(k-i)$ steps, we change the initial matrix to the following form,
  \begin{equation*}
    \prod_{t\in[i],j\in[k-i]}(a_{u_j}-a_{v_t})[C(x^{k-1})\; C(x^{k-2})\;\cdots\;C(1)]^T.
  \end{equation*}
  Since all the row operations in above steps do not change the determinant, we complete the proof.
\end{proof}

\bibliographystyle{IEEEtran}
\bibliography{balanced}


\end{document}